\documentclass[12pt,a4paper]{article}
\usepackage[top=2.3cm,bottom=2.7cm,left=2.2cm,right=2.2cm]{geometry}
\usepackage{latexsym,amsmath,amsfonts,amssymb,amsthm,amscd,mathrsfs}

\newcommand{\be}{\begin{equation}}
\newcommand{\ee}{\end{equation}}
\newcommand{\bea}{\begin{eqnarray}}
\newcommand{\eea}{\end{eqnarray}}

\numberwithin{equation}{section}
\newcounter{thmcounter}
\numberwithin{thmcounter}{section}

\theoremstyle{definition}
\newtheorem*{acknowledgements}{Acknowledgements}
\newtheorem{definition}[thmcounter]{Definition}
\newtheorem{remark}[thmcounter]{Remark}
\theoremstyle{plain}
\newtheorem{corollary}[thmcounter]{Corollary}
\newtheorem{lemma}[thmcounter]{Lemma}
\newtheorem{proposition}[thmcounter]{Proposition}
\newtheorem{theorem}[thmcounter]{Theorem}

\def\BC{\mathrm{BC}}                        %
\def\1{{\boldsymbol 1}}                     %
\def\0{{\boldsymbol 0}}                     %
\def\cC{{\mathcal C}}                       %
\def\cD{{\mathcal D}}                       %
\def\cG{{\mathcal G}}                       %
\def\cH{{\mathcal H}}                       %
\def\cL{{\mathcal L}}                       %
\def\cP{{\mathcal P}}                       %
\def\cZ{{\mathcal Z}}                       %
\def\tr{\mathrm{tr}}                        %
\def\diag{\mathrm{diag}}                    %
\def\ri{{\rm i}}                            %
\def\sgn{{\rm sgn}}                         %
\def\C{\mathbb{C}}                          %
\def\D{\mathbb{D}}                          %
\def\R{\mathbb{R}}                          %
\def\T{\mathbb{T}}                          %
\def\SL{{\rm SL}}                           %
\def\SB{{\rm SB}}                           %
\def\UN{{\rm U}}                            %
\def\SU{{\rm SU}}                           %
\def\su{\mathfrak{su}}                      %
\def\fH{\mathfrak{H}}                       %
\def\sV{{\mathsf V}}                        %
\def\sv{{\mathsf v}}                        %
\def\sw{{\mathsf w}}                        %
\def\ds{\left.\frac{d}{ds}\right\vert_{s=0}}%

\begin{document}
\begin{center}
{\large\bf
On a Poisson-Lie deformation of the $\boldsymbol{\BC_n}$ Sutherland system}
\end{center}

\medskip
\begin{center}
L.~Feh\'er${}^{a,b}$ and T.F.~G\"orbe${}^a$\\

\bigskip
${}^a$Department of Theoretical Physics, University of Szeged\\
Tisza Lajos krt 84-86, H-6720 Szeged, Hungary\\
e-mail: tfgorbe@physx.u-szeged.hu

\medskip
${}^b$Department of Theoretical Physics, WIGNER RCP, RMKI\\
H-1525 Budapest, P.O.B.~49, Hungary\\
e-mail: lfeher@physx.u-szeged.hu
\end{center}

\medskip
\begin{abstract}
A deformation of the classical trigonometric $\BC_n$ Sutherland system is
derived via Hamiltonian reduction of the Heisenberg double of $\SU(2n)$.
We apply a natural Poisson-Lie analogue of the Kazhdan-Kostant-Sternberg type
reduction of the free particle on $\SU(2n)$ that leads to the $\BC_n$ Sutherland
system. We prove that this yields a Liouville integrable Hamiltonian system
and construct a globally valid model of the smooth reduced phase space wherein the
commuting flows are complete. We point out that the reduced system, which contains
3 independent coupling constants besides the deformation parameter, can be recovered
(at least on a dense submanifold) as a singular limit of the standard 5-coupling deformation
due to van Diejen. Our findings complement and further develop those obtained recently
by Marshall on the hyperbolic case by reduction of the Heisenberg double of $\SU(n,n)$.
\end{abstract}

{\linespread{0.8}\tableofcontents}

\newpage
\section{Introduction}
\label{sec:1}

Models amenable to exact treatment provide key paradigms for our understanding of
natural phenomena and form a fertile field of research crossing the border of physics
and mathematics.
The study of integrable Hamiltonian systems is a very active
subfield with particularly strong ties to group theory and symplectic geometry.
For reviews, see e.g.~\cite{FT,N,RuijR,BBT,E}.
One of the time-honoured approaches to such systems consists in viewing them as
`shadows' of natural free systems enjoying high symmetries. This is alternatively
known as the projection method or as Hamiltonian reduction \cite{OP1,OP2}. The list of the
free `master systems' is monotonically expanding in time. To name a few, it includes
free particles on Lie groups together with their Poisson-Lie symmetric deformations
and quasi-Hamiltonian analogues. For example, it was shown in the pioneering paper
\cite{KKS} that the integrable many-body system of Sutherland \cite{Sut}, which describes
particles on the circle interacting via a pair potential given by the inverse square
of the chord-distance, is a reduction of the free particle on the unitary group
$\UN(n)$. Various deformations of the Sutherland system due to Ruijsenaars and
Schneider \cite{RS86,RIMS95} were derived \cite{FK1,FK2} from Poisson-Lie symmetric
free motion on $\UN(n)$, whose phase space is the Heisenberg double \cite{STS} of
the Poisson-Lie group $\UN(n)$, and from the internally fused quasi-Hamiltonian
 double \cite{AMM} of $\UN(n)$, which arose from Chern-Simons field theory.

The projection method was enriched by an interesting recent contribution of Marshall
\cite{M}, who obtained an integrable Ruijsenaars-Schneider (RS) type system by reducing the
Heisenberg double of $\SU(n,n)$, which directly motivated our present
work\footnote{The relation is `symmetric' as the problem studied by Marshall was
originally suggested by one of us.}. Here, we shall deal with a reduction of the
Heisenberg double of $\SU(2n)$ and derive a Liouville integrable Hamiltonian system
related to Marshall's one in a way similar to the connection between the original
trigonometric Sutherland system and its hyperbolic variant. Although this is
essentially analytic continuation, it should be noted that the resulting systems are
qualitatively different in their dynamical characteristics and global features. In
addition, what we hope makes our work worthwhile is that our treatment is different
from the one in \cite{M} in several respects and we go considerably further regarding
the global characterization of the reduced phase space and the completeness
of the relevant Hamiltonian flows.

The main Hamiltonian of the system that we obtain can be displayed as follows
\begin{multline}
H(\hat p,\hat q;x,u,v)=\frac{e^{-2u}+e^{2v}}{2}\sum_{j=1}^ne^{-2\hat p_j}+\\
\qquad-\sum_{j=1}^n\cos(\hat q_j)\big[1-(1+e^{2(v-u)})e^{-2\hat p_j}
+e^{2(v-u)}e^{-4\hat p_j}\big]^{\tfrac{1}{2}}\prod_{\substack{k=1\\(k\neq j)}}^n
\bigg[1-\frac{\sinh^2\big(\frac{x}{2}\big)}{\sinh^2(\hat p_j-\hat p_k)}
\bigg]^{\tfrac{1}{2}}.
\label{I1}
\end{multline}
Here $u$, $v$ and $x$ are real coupling parameters that will be assumed to satisfy
\be
u<v,\quad v\neq-u\quad\text{and}\quad x\neq 0.
\label{I2}
\ee
The components of $\hat q$ parametrize the torus $\T_n$ by $e^{\ri\hat q}$ and
$\hat p$ belongs to the domain
\be
\cC_x:=\{\hat p\in\R^n\,\mid
0>\hat p_1,\ \hat p_k-\hat p_{k+1}>|x|/2\ (k=1,\dots,n-1)\}.
\label{I3}
\ee
The dynamics is then defined via the symplectic form
\be
\hat\omega=\sum_{j=1}^nd\hat q_j\wedge d\hat p_j.
\label{I4}
\ee
It will be shown that this system results by restricting a reduced free system
on a dense open submanifold of the pertinent reduced phase space. The Hamiltonian
flow is complete on the full reduced phase space, but it can leave the submanifold
parametrized by $\cC_x\times\T_n$. By glancing at the form of the Hamiltonian, one
may say that it represents an RS type system coupled to external fields. Since
differences of the `position variables' $\hat p_k$ appear, one feels that this
Hamiltonian somehow corresponds to an A-type root system.

To better understand the nature of this model, let us now introduce new Darboux
variables $q_k$, $p_k$ following essentially \cite{M} as
\be
\exp(\hat p_k)=\sin(q_k)\quad\text{and}\quad\hat q_k=p_k\tan(q_k).
\label{I5}
\ee
In terms of these variables $H(\hat p,\hat q;x,u,v)=\cH_1(q,p;x,u,v)$ with the
`new Hamiltonian'
\begin{multline}
\cH_1(q,p; x,u,v)=\frac{e^{-2u}+e^{2v}}{2}\sum_{j=1}^n\frac{1}{\sin^2(q_j)}\\
-\sum_{j=1}^n\cos(p_j\tan(q_j))\bigg[1-\frac{1+e^{2(v-u)}}{\sin^2(q_j)}
+\frac{4e^{2(v-u)}}{4\sin^2(q_j)-\sin^2(2q_j)}\bigg]^{\tfrac{1}{2}}\\
\times\prod_{\substack{k=1\\(k\neq j)}}^n
\bigg[1-\frac{2\sinh^2\big(\frac{x}{2}\big)\sin^2(q_j)\sin^2(q_k)}
{\sin^2(q_j-q_k)\sin^2(q_j+q_k)}
\bigg]^{\tfrac{1}{2}}.
\label{I6}
\end{multline}
Remarkably, only such combinations of the new `position variables' $q_k$ appear that
are naturally associated with the $\BC_n$ root system and the Hamiltonian $\cH_1$
enjoys symmetry under the corresponding Weyl group. Thus now one may wish to attach
the Hamiltonian $\cH_1$ to the $\BC_n$ root system. Indeed, this interpretation is
preferable for the following reason. Introduce the scale parameter (corresponding to
the inverse of the velocity of light in the original RS system) $\beta>0$ and make
the substitutions
\be
u\to\beta u,\quad
v\to\beta v,\quad
x\to\beta x,\quad
p\to\beta p,\quad
\hat\omega\to\beta\hat\omega.
\label{I7}
\ee
Then consider the deformed Hamiltonian
\be
\cH_\beta(q,p;x,u,v):=\cH_1(q,\beta p;\beta x,\beta u,\beta v).
\label{I8}
\ee
The point is that one can then verify the following relation:
\be
\lim_{\beta\to 0}\frac{\cH_\beta(q,p;x,u,v)-n}{\beta^2}
=H_{\BC_n}^{\text{Suth}}(q,p;\gamma,\gamma_1,\gamma_2),
\label{I9}
\ee
where
\be
H_{\BC_n}^{\text{Suth}}=\frac{1}{2}\sum_{j=1}^np_j^2
+\sum_{1\leq j<k\leq n}\bigg[\frac{\gamma}{\sin^2(q_j-q_k)}
+\frac{\gamma}{\sin^2(q_j+q_k)}\bigg]
+\sum_{j=1}^n \frac{\gamma_1}{\sin^2(q_j)}
+\sum_{j=1}^n \frac{\gamma_2}{\sin^2(2q_j)}
\label{I10}
\ee
is the standard trigonometric $\BC_n$ Sutherland Hamiltonian with coupling constants
\be
\gamma=\frac{x^2}{4},\quad
\gamma_1=2uv,\quad
\gamma_2=2(v-u)^2.
\label{I11}
\ee
Note that the domain of the variables $\hat q,\hat p$, and correspondingly that of
$q$, $p$ also depends on $\beta$, and in the $\beta \to 0$ limit it is easily seen
that we recover the usual $\BC_n$ domain
\be
\frac{\pi}{2}>q_1>q_2>\dots>q_n>0,\quad p\in\R^n.
\label{I12}
\ee
In conclusion, we see that $H$ in its equivalent form $\cH_\beta$ is a 1-parameter
deformation of the trigonometric $\BC_n$ Sutherland Hamiltonian. We remark in
passing that the conditions \eqref{I2} imply that $\gamma_2>0$ and
$4\gamma_1+\gamma_2>0$, which guarantee that the flows of $H_{\BC_n}^{\text{Suth}}$
are complete on the domain \eqref{I12}.

Marshall \cite{M} obtained similar results for an analogous deformation of the
hyperbolic $\BC_n$ Sutherland Hamiltonian. His deformed Hamiltonian differs from
\eqref{I1} above in some important signs and in the relevant domain of the `position
variables' $\hat p$. Although in our impression the completeness of the reduced
Hamiltonian flows was not treated in a satisfactory way in \cite{M}, the
completeness proof that we shall present can be adapted to Marshall's case as well.

It is natural to ask how the system studied in the present paper (and its cousin in
\cite{M}) is related to van Diejen's \cite{vD} 5-coupling trigonometric $\BC_n$
system? It was shown already in \cite{vD} that the 5-coupling trigonometric system
is a deformation of the $\BC_n$ Sutherland system, and later \cite{vD2} several other
integrable systems were also derived as its (`Inozemtsev type' \cite{I}) limits.
Motivated by this, we can show that the Hamiltonian \eqref{I1} is a
singular\footnote{We call the limit singular since it involves sending some shifted
position variables to infinity.} limit of van Diejen's general Hamiltonian.
Incidentally, a Hamiltonian of Schneider \cite{S} can be viewed as a subsequent singular
limit of the Hamiltonian \eqref{I1}. Schneider's system was mentioned in \cite{M},
too, but the relation to van Diejen's system was not described.

The original idea behind the present work and \cite{M} was that a natural Poisson-Lie
analogue of the Hamiltonian reduction treatment \cite{FP} of the $\BC_n$
Sutherland system should lead to a deformation of this system. It was expected that a
special case of van Diejen's standard 5-coupling deformation will arise. The
expectation has now been confirmed, although it came as a surprise that a singular
limit is involved in the connection.

The outline of the paper is as follows. We start in Section \ref{sec:2} by defining
the reduction of interest. In Section \ref{sec:3} we observe that several technical
results of \cite{FK1} can be applied for analyzing the reduction at hand, and solve
the momentum map constraints by taking advantage of this observation. The heart of
the paper is Section \ref{sec:4}, where we characterize the reduced system.
In Subsection \ref{subsec:4.1} we prove that the reduced phase space is smooth, as
formulated in Theorem \ref{thm:4.4}. Then in Subsection \ref{subsec:4.2} we focus on a dense
open submanifold on which the Hamiltonian \eqref{I1} lives. The demonstration of the
Liouville integrability of the reduced free flows is given in Subsection
\ref{subsec:4.3}. In particular, we prove the integrability of the completion of the
system \eqref{I1} carried by the full reduced phase space. Our main result is Theorem
\ref{thm:4.9} (proved in Subsection \ref{subsec:4.4}), which establishes a globally valid model
of the reduced phase space. We stress that the global structure of the phase space on
which the flow of \eqref{I1} is complete was not considered previously at all, and
will be clarified as a result of our group theoretic interpretation. Section
\ref{sec:5} contains our conclusions, further comments on the related paper by
Marshall \cite{M} and a discussion of open problems. The main text is complemented by
four appendices. Appendix \ref{sec:A} deals with the connection to van Diejen's system;
the other 3 appendices contain important details relegated from the main text.

\section{Definition of the Hamiltonian reduction}
\label{sec:2}

We below introduce the `free' Hamiltonians and define their reduction. We restrict
the presentation of this background material to a minimum necessary for
understanding our work. The conventions follow \cite{FK1}, which also contains
more details. As a general reference, we recommend \cite{CP}.

\subsection{The unreduced free Hamiltonians}
\label{subsec:2.1}

We fix a natural number\footnote{The $n=1$ case would need special treatment
and is excluded in order to simplify the presentation.} $n\geq 2$ and
consider the Lie group $\SU(2n)$ equipped with its standard quadratic Poisson bracket
defined by the compact form of the Drinfeld-Jimbo classical $r$-matrix,
\be
r_{\mathrm{DJ}}=\ri\sum_{1\leq\alpha<\beta\leq 2n}E_{\alpha\beta}\wedge E_{\beta\alpha},
\label{T1}
\ee
where $E_{\alpha\beta}$ is the elementary matrix of size $2n$ having a single
non-zero entry $1$ at the $\alpha\beta$ position. In particular, the Poisson
brackets of the matrix elements of $g\in\SU(2n)$ obey Sklyanin's formula
\be
\{g\stackrel{\otimes}{,}g\}_{\SU(2n)}=[g\otimes g, r_{\mathrm{DJ}}].
\label{T2}
\ee
Thus $\SU(2n)$ becomes a Poisson-Lie group, i.e., the multiplication
$\SU(2n)\times\SU(2n)\to\SU(2n)$ is a Poisson map. The cotangent bundle $T^\ast\SU(2n)$
possesses a natural Poisson-Lie analogue, the so-called Heisenberg double \cite{STS},
which is provided by the real Lie group $\SL(2n,\C)$ endowed with a certain
symplectic form \cite{AM},
$\omega$. To describe $\omega$, we use the Iwasawa decomposition and factorize every
element $K\in\SL(2n,\C)$ in two alternative ways
\be
K=g_Lb_R^{-1}=b_Lg_R^{-1}
\label{T3}
\ee
with uniquely determined
\be
g_L,g_R\in\SU(2n),\quad
b_L,b_R\in\SB(2n).
\label{T4}
\ee
Here $\SB(2n)$ stands for the subgroup of $\SL(2n,\C)$ consisting of upper triangular
matrices with positive diagonal entries. The symplectic form $\omega$ reads
\be
\omega=\frac{1}{2}\Im\tr(db_Lb_L^{-1}\wedge dg_Lg_L^{-1})+
\frac{1}{2}\Im\tr(db_Rb_R^{-1}\wedge dg_Rg_R^{-1}).
\label{T5}
\ee
Before specifying free Hamiltonians on the phase space $\SL(2n,\C)$, note that any
smooth function $h$ on $\SB(2n)$ corresponds to a function $\tilde h$ on the space of
positive definite Hermitian matrices of determinant $1$ by the relation
\be
\tilde h(bb^\dagger)=h(b),\quad\forall b\in\SB(2n).
\label{T6}
\ee
Then introduce the invariant functions
\be
C^\infty(\SB(2n))^{\SU(2n)}\equiv\{h\in C^\infty(\SB(2n))\mid\tilde h(bb^\dagger)
=\tilde h(gbb^\dagger g^{-1}),\ \forall g\in\SU(2n),b\in\SB(2n)\}.
\label{T7}
\ee
These in turn give rise to the following ring of functions on $\SL(2n,\C)$:
\be
\fH\equiv\{\cH\in C^\infty(\SL(2n,\C))\mid\cH(g_Lb_R^{-1})=h(b_R),\
h\in C^{\infty}(\SB(2n))^{\SU(2n)}\},
\label{T8}
\ee
where we utilized the decomposition \eqref{T3}. An important point is that $\fH$
forms an Abelian algebra with respect to the Poisson bracket associated with $\omega$
\eqref{T5}.

The flows of the `free' Hamiltonians contained in $\fH$ can be obtained effortlessly.
To describe the result, define the derivative $d^Rf\in C^\infty(\SB(2n),\su(2n))$ of any
real function $f\in C^\infty(\SB(2n))$ by requiring
\be
\ds f(be^{sX})=\Im\tr\big(Xd^Rf(b)\big),
\quad\forall b\in\SB(2n),\ \forall X\in\mathrm{Lie}(\SB(2n)).
\label{T9}
\ee
The Hamiltonian flow generated by $\cH\in\fH$ through the initial value
$K(0)=g_L(0)b_R(0)^{-1}$ is in fact given by
\be
K(t)=g_L(0)\exp\big[-td^Rh(b_R(0))\big]b_R^{-1}(0),
\label{T10}
\ee
where $\cH$ and $h$ are related according to \eqref{T8}. This means that $g_L(t)$
follows the orbit of a one-parameter subgroup, while $b_R(t)$ remains constant.
Actually, $g_R(t)$ also varies along a similar orbit, and $b_L(t)$ is constant.

The constants of motion $b_L$ and $b_R$ generate a Poisson-Lie symmetry, which
allows one to define Marsden-Weinstein type \cite{MW} reductions.

\subsection{Generalized Marsden-Weinstein reduction}
\label{subsec:2.2}

The free Hamiltonians in $\fH$ are invariant with respect to the action of
$\SU(2n)\times\SU(2n)$ on $\SL(2n,\C)$ given by left- and right-multiplications.
This is a Poisson-Lie symmetry, which means that the corresponding action map
\be
\SU(2n)\times\SU(2n)\times\SL(2n,\C)\to\SL(2n,C),
\label{T11}
\ee
operating as
\be
(\eta_L,\eta_R,K)\mapsto\eta_L K\eta_R^{-1},
\label{T12}
\ee
is a Poisson map. In \eqref{T11} the product Poisson structure is taken using the
Sklyanin bracket on $\SU(2n)$ and the Poisson structure on $\SL(2n,\C)$ associated with
the symplectic form $\omega$ \eqref{T5}. This Poisson-Lie symmetry admits a momentum
map in the sense of Lu \cite{Lu}, given explicitly by
\be
\Phi\colon\SL(2n,\C)\to\SB(2n)\times\SB(2n),\quad\Phi(K)=(b_L,b_R).
\label{T13}
\ee
The key property of the momentum map is represented by the identity
\be
\ds f(e^{sX}Ke^{-sY})=\Im\tr\big(X\{f,b_L\}b_L^{-1}+Y\{f,b_R\}b_R^{-1}\big),
\quad\forall X,Y\in\su(2n),
\label{T14}
\ee
where $f\in C^\infty(\SL(2n,\C))$ is an arbitrary real function and the Poisson bracket
is the one corresponding to $\omega$ \eqref{T5}. The map $\Phi$ enjoys an
equivariance property and one can \cite{Lu} perform Marsden-Weinstein type reduction in the
same way as for usual Hamiltonian actions (for which the symmetry group has vanishing
Poisson structure). To put it in a nutshell, any $\cH\in\fH$ gives rise to a reduced
Hamiltonian system by fixing the value of $\Phi$ and subsequently  taking quotient
with respect to the
corresponding isotropy group. The reduced flows can be obtained by the standard
restriction-projection algorithm, and under favorable circumstances the reduced phase
space is a smooth symplectic manifold.

Now, consider the block-diagonal subgroup
\be
G_+:=\mathrm{S}(\UN(n)\times\UN(n))<\SU(2n).
\label{T15}
\ee
Since $G_+$ is also a Poisson submanifold of $\SU(2n)$, the restriction of \eqref{T12}
yields a Poisson-Lie action
\be
G_+\times G_+\times\SL(2n,\C)\to\SL(2n,\C)
\label{T16}
\ee
of $G_+\times G_+$. The momentum map for this action is provided by projecting the
original momentum map $\Phi$ as follows. Let us write every element $b\in\SB(2n)$ in
the block-form
\be
b=\begin{bmatrix}b(1)&b(12)\\\0_n&b(2)\end{bmatrix}
\label{T17}
\ee
and define $G_+^\ast<\SB(2n)$ to be the subgroup for which $b(12)=\0_n$.
If $\pi\colon\SB(2n)\to G_+^\ast$ denotes the projection
\be
\pi\colon\begin{bmatrix}b(1)&b(12)\\\0_n&b(2)\end{bmatrix}\mapsto
\begin{bmatrix}b(1)&\0_n\\\0_n&b(2)\end{bmatrix},
\label{T18}
\ee
then the momentum map $\Phi_+\colon\SL(2n,\C)\to G_+^\ast\times G_+^\ast$ is furnished
by
\be
\Phi_+(K)=(\pi(b_L),\pi(b_R)).
\label{T19}
\ee
Indeed, it is readily checked that the analogue of \eqref{T14} holds with $X$, $Y$
taken from the block-diagonal subalgebra of $\su(2n)$ and $b_L$, $b_R$ replaced by
their projections. The equivariance property of this momentum map means that
in correspondence to
\be
K\mapsto\eta_L K\eta_R^{-1}\quad\text{with}\quad(\eta_L,\eta_R)\in G_+\times G_+,
\label{T20}
\ee
one has
\be
\big(\pi(b_L)\pi(b_L)^\dagger,\pi(b_R)\pi(b_R)^\dagger\big)\mapsto
\big(\eta_L\pi(b_L)\pi(b_L)^\dagger\eta_L^{-1},
\eta_R\pi(b_R)\pi(b_R)^\dagger \eta_R^{-1}\big).
\label{T21}
\ee
We briefly mention here that, as the notation suggests, $G_+^\ast$ is itself a
Poisson-Lie group that can serve as a Poisson dual of $G_+$. The relevant Poisson
structure can be obtained by identifying the block-diagonal subgroup of $\SB(2n)$ with
the factor group $\SB(2n)/L$, where $L$ is the block-upper-triangular normal subgroup.
This factor group inherits a Poisson structure from $\SB(2n)$, since $L$ is a so-called
coisotropic (or `admissible') subgroup of $\SB(2n)$ equipped with its standard Poisson
structure. The projected momentum map $\Phi_+$ is a Poisson map with respect to this
Poisson structure on the two factors $G_+^\ast$ in \eqref{T19}. The details are not
indispensable for us. The interested reader may find them e.g. in \cite{BCST}.

Inspired by the papers \cite{FP,FK1,M}, we wish to study the particular
Marsden-Weinstein reduction defined by imposing the following momentum map constraint:
\be
\Phi_+(K)=\mu\equiv(\mu_L,\mu_R),\quad\text{where}\quad
\mu_L=\begin{bmatrix}e^u\nu(x)&\0_n\\\0_n&e^{-u}\1_n\end{bmatrix},\quad
\mu_R=\begin{bmatrix}e^v\1_n&\0_n\\\0_n&e^{-v}\1_n\end{bmatrix}
\label{T22}
\ee
with some real constants $u$, $v$ and $x$. Here, $\nu(x)\in\SB(n)$ is the $n\times n$
upper triangular matrix defined by
\be
\nu(x)_{jj}=1,\quad\nu(x)_{jk}=(1-e^{-x})e^{\frac{(k-j)x}{2}},\quad j<k,
\label{T23}
\ee
whose main property is that $\nu(x)\nu(x)^\dag$ has the largest possible
non-trivial isotropy group under conjugation by the elements of $\SU(n)$.

Our principal task is to characterize the reduced phase space
\be
M\equiv\Phi_+^{-1}(\mu)/G_\mu,
\label{T24}
\ee
where $\Phi_+^{-1}(\mu)=\{K\in\SL(2n,\C)\mid\Phi_+(K)=\mu\}$
and
\be
G_\mu=G_+(\mu_L)\times G_+
\label{T25}
\ee
is the isotropy group of $\mu$ inside $G_+\times G_+$. Concretely, $G_+(\mu_L)$ is
the subgroup of $G_+$ consisting of the special unitary matrices of the form
\be
\eta_L=\begin{bmatrix}\eta_L(1)&\0_n\\\0_n&\eta_L(2)\end{bmatrix},
\label{T26}
\ee
where $\eta_L(2)$ is arbitrary and
\be
\eta_L(1)\nu(x)\nu(x)^\dag\eta_L(1)^{-1}=\nu(x)\nu(x)^\dag.
\label{T27}
\ee
In words, $\eta_L(1)$ belongs to  the little group of $\nu(x)\nu(x)^\dag$ in $\UN(n)$.
We shall see that $\Phi_+^{-1}(\mu)$ and $M$ are smooth manifolds for which the
canonical projection
\be
\pi_\mu\colon\Phi_+^{-1}(\mu)\to M
\label{T28}
\ee
is a smooth submersion. Then $M$ \eqref{T24} inherits a symplectic form
$\omega_M$ from $\omega$ \eqref{T5}, which satisfies
\be
\iota_\mu^\ast(\omega)=\pi_\mu^\ast(\omega_M),
\label{T29}
\ee
where $\iota_\mu\colon\Phi_+^{-1}(\mu)\to\SL(2n,\C)$ denotes the tautological
embedding.

\section{Solution of the momentum map constraints}
\label{sec:3}

The description of the reduced phase space requires us to solve the momentum map
constraints, i.e., we have to find all elements $K\in\Phi_+^{-1}(\mu)$. Of course,
it is enough to do this up to the gauge transformations provided by the isotropy
group $G_\mu$ \eqref{T25}. The solution of this problem will rely on the auxiliary
equation \eqref{S11} below, which is essentially equivalent to the momentum map
constraint, $\Phi_+(K)=\mu$, and coincides with an equation studied previously in
great detail in \cite{FK1}. Thus we start in the next subsection by deriving this
equation.

\subsection{A crucial equation implied by the constraints}
\label{subsec:3.1}

We begin by recalling (e.g.~\cite{Mat}) that any $g\in\SU(2n)$ can be decomposed as
\be
g=g_+\begin{bmatrix}\cos q&\ri\sin q\\\ri\sin q&\cos q\end{bmatrix}h_+,
\label{S1}
\ee
where $g_+,h_+\in G_+$ and $q=\diag(q_1,\ldots,q_n)\in\R^n$ satisfies
\be
\frac{\pi}{2}\geq q_1\geq\dots\geq q_n\geq 0.
\label{S2}
\ee
The vector $q$ is uniquely determined by $g$, while $g_+$ and $h_+$ suffer from
controlled ambiguities.

First, apply the above decomposition to $g_L$ in $K=g_Lb_R^{-1}\in\Phi_+^{-1}(\mu)$
and use the right-handed momentum constraint $\pi(b_R)=\mu_R$. It is then easily seen
that up to gauge transformations every element of $\Phi_+^{-1}(\mu)$ can be
represented in the following form:
\be
K=\begin{bmatrix}\rho&\0_n\\\0_n&\1_n\end{bmatrix}
\begin{bmatrix}\cos q&\ri\sin q\\\ri\sin q&\cos q\end{bmatrix}
\begin{bmatrix}e^{-v}\1_n&\alpha\\\0_n&e^v\1_n\end{bmatrix}.
\label{S3}
\ee
Here $\rho\in\SU(n)$ and $\alpha$ is an $n\times n$ complex matrix. By using obvious
block-matrix notation, we introduce $\Omega:=K_{22}$ and record from \eqref{S3} that
\be
\Omega=\ri(\sin q)\alpha+e^v\cos q.
\label{S4}
\ee
For later purpose we introduce also the polar decomposition of the matrix $\Omega$,
\be
\Omega=\Lambda T,
\label{S5}
\ee
where $T\in\UN(n)$ and the Hermitian, positive semi-definite factor $\Lambda$ is
uniquely determined by the relation $\Omega\Omega^\dag=\Lambda^2$.

Second, by writing $K=b_Lg_R^{-1}$ the left-handed momentum constraint
$\pi(b_L) =\mu_L$ tells us that $b_L$ has the block-form
\be
b_L=\begin{bmatrix}e^u\nu(x)& \chi\\\0_n&e^{-u}\1_n\end{bmatrix}
\label{S6}
\ee
with an $n\times n$ matrix $\chi$. Now we inspect the components of the $2\times 2$
block-matrix identity
\be
K K^\dag=b_Lb_L^\dag,
\label{S7}
\ee
which results by substituting $K$ from \eqref{S3}. We find that the (22) component of
this identity is equivalent to
\be
\Omega\Omega^\dag=\Lambda^2=e^{-2u}\1_n-e^{-2v}(\sin q)^2.
\label{S8}
\ee
On account of the condition \eqref{I2}, this uniquely determines $\Lambda$ in terms
of $q$, and shows also that $\Lambda$ is invertible. A further important consequence
is that we must have
\be
q_n>0,
\label{S9}
\ee
and therefore $\sin q$ is an invertible diagonal matrix. Indeed, if $q_n=0$, then
from \eqref{S4} and \eqref{S8} we would get
$(\Omega\Omega^\dag)_{nn}=e^{2v}=e^{-2u}$, which is excluded by \eqref{I2}.

Next, one can check that in the presence of the relations already established, the
(12) and the (21) components of the identity \eqref{S7} are equivalent to the equation
\be
\chi=\rho(\ri\sin q)^{-1}[e^{-u}\cos q-e^{u+v}\Omega^\dag].
\label{S10}
\ee
Observe that $K$ uniquely determines $q$, $T$ and $\rho$, and conversely $K$ is
uniquely defined by the above relations once $q$, $T$ and $\rho$ are found.

Now one can straightforwardly check by using the above relations that the (11)
component of the identity \eqref{S7} translates into the following equation:
\be
\rho(\sin q)^{-1}T^\dag(\sin q)^2T(\sin q)^{-1}\rho^\dag=\nu(x)\nu(x)^\dag.
\label{S11}
\ee
This is to be satisfied by $q$ subject to \eqref{S2}, \eqref{S9} and $T\in\UN(n)$,
$\rho\in\SU(n)$. What makes our job relatively easy is that this is the same as
equation (5.7) in the paper \cite{FK1} by Klim\v c\'ik and one of us. In fact,
this equation
was analyzed in  detail in \cite{FK1}, since it played a crucial role in that work,
too. The correspondence with the symbols used in \cite{FK1} is
\be
(\rho,T,\sin q)\Longleftrightarrow(k_L,k_R^\dag,e^{\hat p}).
\label{S12}
\ee
This motivates to introduce the variable $\hat p\in\R^n$ in our case, by setting
\be
\sin q_k=e^{\hat p_k},\quad k=1,\dots,n.
\label{S13}
\ee
Notice from \eqref{S2} and \eqref{S9} that we have
\be
0\geq\hat p_1\geq\dots\geq\hat p_n>-\infty.
\label{S14}
\ee
If the components of $\hat p$ are all different, then we can directly rely on \cite{FK1} to
establish both the allowed range of $\hat p$ and the explicit form of $\rho$ and $T$.
The statement that $\hat p_j\neq\hat p_k$ holds for $j\neq k$ can be proved by
adopting arguments given in \cite{FK1,FK2}. This proof requires combining techniques of
\cite{FK1} and \cite{FK2}, whose extraction from \cite{FK1,FK2} is rather involved.
 We present it in Appendix \ref{sec:B}, otherwise in the next subsection we proceed by simply
stating relevant applications of results from \cite{FK1}.

\begin{remark}\label{rem:3.1}
In the context of \cite{FK1} the components of $\hat p$ are not restricted to the
half-line and both $k_L$ and $k_R$ vary in $\UN(n)$. These slight differences do
not pose any obstacle to using the results and techniques of \cite{FK1,FK2}. We note that
essentially the same equation \eqref{S11} surfaced in \cite{M} as well, but the author of
that paper refrained from taking advantage of the previous analyses of this equation.
In fact, some statements of \cite{M} are not fully correct. This will be specified (and
corrected) in Section \ref{sec:5}.
\end{remark}

\subsection{Consequences of equation \eqref{S11}}
\label{subsec:3.2}

We start by pointing out the foundation of the whole analysis. For this, we
first display the identity
\be
\nu(x)\nu(x)^\dag=e^{-x}\1_n+\sgn(x)\hat v\hat v^\dag,
\label{S15}
\ee
which holds with a certain $n$-component vector $\hat v=\hat v(x)$. By introducing
\be
w=\rho^\dag\hat v
\label{S16}
\ee
and setting $\hat p\equiv\diag(\hat p_1,\dots,\hat p_n)$, we rewrite equation
\eqref{S11} as
\be
e^{2\hat p-x\1_n}+\sgn(x)e^{\hat p}ww^\dag e^{\hat p}=T^{-1}e^{2\hat p}T.
\label{S17}
\ee
The equality of the characteristic polynomials of the matrices on the two sides of
\eqref{S17} gives a polynomial equation that contains $\hat p$, the absolute values
$|w_j|^2$ and a complex indeterminate. Utilizing the requirement that $|w_j|^2\geq 0$
must hold, one obtains the following result.

\begin{proposition}\label{prop:3.2}
If $K$ given by \eqref{S3} belongs to the constraint
surface $\Phi_+^{-1}(\mu)$, then
the vector $\hat p$ \eqref{S13} is contained in the closed polyhedron
\be
\bar\cC_x:=\{\hat p\in\R^n\mid 0\geq\hat p_1,\
\hat p_k-\hat p_{k+1}\geq|x|/2\ (k=1,\dots,n-1)\}.
\label{S18}
\ee
\end{proposition}\noindent
Proposition \ref{prop:3.2} can be proved by merging the proofs of Lemma 5.2
of \cite{FK1} and Theorem 2 of \cite{FK2}. This is presented in Appendix \ref{sec:B}.

The above-mentioned polynomial equality permits to find the possible vectors $w$ \eqref{S16}
as well. If $\hat p$ and $w$ are given, then $T$ is determined by equation \eqref{S17}
up to left-multiplication by a diagonal matrix and $\rho$ is determined by
\eqref{S16} up to left-multiplication by elements from the little group of
$\hat v(x)$. Following this line of reasoning and controlling the ambiguities in the
same way as in \cite{FK1}, one can find the explicit form of the most general $\rho$ and
$T$ at \emph{any} fixed $\hat p\in\bar\cC_x$. In particular, it turns out that the
range of the vector $\hat p$ equals $\bar \cC_x$.

Before presenting the result, we need to prepare some notations. First of all, we
pick an arbitrary $\hat p\in\bar\cC_x$ and define the $n\times n$ matrix
$\theta(x,\hat p)$ as follows:
\be
\theta(x,\hat p)_{jk}:=\frac{\sinh\big(\frac{x}{2}\big)}{\sinh(\hat p_k-\hat p_j)}
\prod_{\substack{m=1\\(m\neq j,k)}}^n\bigg[\frac{\sinh(\hat p_j-\hat p_m-\frac{x}{2})
\sinh(\hat p_k-\hat p_m+\frac{x}{2})}{\sinh(\hat p_j-\hat p_m)
\sinh(\hat p_k-\hat p_m)}\bigg]^{\tfrac{1}{2}},\quad j\neq k,
\label{S19}
\ee
and
\be
\theta(x,\hat p)_{jj}:=\prod_{\substack{m=1\\(m\neq j)}}^n
\bigg[\frac{\sinh(\hat p_j-\hat p_m-\frac{x}{2})\sinh(\hat p_j-\hat p_m+
\frac{x}{2})}{\sinh^2(\hat p_j-\hat p_m)}\bigg]^{\tfrac{1}{2}}.
\label{S20}
\ee
All expressions under square root are non-negative and non-negative square roots are
taken. Note that $\theta(x,\hat p)$ is a real orthogonal matrix of determinant 1 for
which $\theta(x,\hat p)^{-1}=\theta(-x,\hat p)$ holds, too.

Next, define the real vector $r(x,\hat p)\in\R^n$ with non-negative components
\be
r(x,\hat p)_j=\sqrt{\frac{1-e^{-x}}{1-e^{-nx}}}\prod_{\substack{k=1\\(k\neq j)}}^n
\sqrt{\frac{1-e^{2\hat p_j-2\hat p_k-x}}{1-e^{2\hat p_j-2\hat p_k}}},
\quad j=1,\dots,n,
\label{S21}
\ee
and the real $n\times n$ matrix $\zeta(x,\hat p)$,
\be
\begin{split}
&\zeta(x,\hat p)_{aa}=r(x,\hat p)_a,\quad
\zeta(x,\hat p)_{ij}=\delta_{ij}-\frac{r(x,\hat p)_ir(x,\hat p)_j}{1+r(x,\hat p)_a},\\
&\zeta(x,\hat p)_{ia}=-\zeta(x,\hat p)_{ai}=r(x,\hat p)_i,\quad i,j\neq a,
\end{split}
\label{S22}
\ee
where $a=n$ if $x>0$ and $a=1$ if $x<0$. Introduce also the vector $v=v(x)$:
\be
v(x)_j=\sqrt{\frac{n(e^x-1)}{1-e^{-nx}}}e^{-\tfrac{jx}{2}},\quad j=1,\ldots,n,
\label{S23}
\ee
which is related to $\hat v$ in \eqref{S15} by
\be
\hat v(x)=\sqrt{\sgn(x) e^{-x}\frac{e^{nx}-1}{n}}v(x).
\label{S24}
\ee
Finally, define the $n\times n$ matrix $\kappa(x)$ as
\be
\begin{gathered}
\kappa(x)_{aa}=\frac{v(x)_a}{\sqrt{n}},\quad
\kappa(x)_{ij}=\delta_{ij}-\frac{v(x)_iv(x)_j}{n+\sqrt{n}v(x)_a},\\
\kappa(x)_{ia}=-\kappa(x)_{ai}=\frac{v(x)_i}{\sqrt{n}},\quad i,j\neq a,
\end{gathered}
\label{S25}
\ee
where, again, $a=n$ if $x>0$ and $a=1$ if $x<0$. It can be shown that both
$\kappa(x)$ and $\zeta(x,\hat p)$ are orthogonal matrices of
determinant 1 for any $\hat p\in\bar\cC_x$.

Now we can state the main result of this section, whose proof is omitted since it is
a direct application of the analysis of the solutions of \eqref{S11} presented in
Section 5 of \cite{FK1}.

\begin{proposition}\label{prop:3.3}
Take any $\hat p\in\bar\cC_x$ and any diagonal unitary matrix $e^{\ri\hat q}\in\T_n$.
By using the preceding notations define $K\in\SL(2n,\C)$ \eqref{S3} by setting
\be
T=e^{\ri\hat q}\theta(-x,\hat p),\quad\rho=\kappa(x)\zeta(x,\hat p)^{-1},
\label{S26}
\ee
and also applying the equations \eqref{S4}, \eqref{S5}, \eqref{S8} and \eqref{S13}.
Then the element $K$ belongs to the constraint surface $\Phi_+^{-1}(\mu)$, and every
orbit of the gauge group $G_\mu$ \eqref{T25} in $\Phi_+^{-1}(\mu)$ intersects the set
of elements $K$ just constructed.
\end{proposition}

\begin{remark}\label{rem:3.4}
It is worth spelling out the expression of the element $K$ given by
Proposition \ref{prop:3.3}. Indeed, we have
\be
K(\hat p, e^{\ri \hat q}) =\begin{bmatrix}\rho &\0_n\\\0_n&\1_n\end{bmatrix}
\begin{bmatrix}\sqrt{\1_n-e^{2\hat p}} &\ri e^{\hat p}\\\ri e^{\hat p}&
\sqrt{\1_n-e^{2\hat p}}\end{bmatrix}
\begin{bmatrix}e^{-v}\1_n&\alpha\\\0_n&e^v\1_n\end{bmatrix}
\label{S27}
\ee
using the above definitions and
\be
\alpha=-\ri\bigg[e^{\ri\hat q}\sqrt{e^{-2u}e^{-2\hat p}-e^{-2v}\1_n}\,\theta(-x,\hat p)
-e^v\sqrt{e^{-2\hat p}-\1_n}\bigg].
\label{S28}
\ee
\end{remark}

\begin{remark}\label{rem:3.5}
Let us call $S$ the set of the elements $K(\hat p,e^{\ri \hat q})$ constructed above,
and observe that this set is homeomorphic to
\be
\bar\cC_x\times\T_n=\{(\hat p,e^{\ri\hat q})\}
\label{S29}
\ee
by its very definition. This is not a smooth manifold, because of the presence of the
boundary of $\bar\cC_x$. However, this does not indicate any `trouble' since it is
not true (at the boundary of $\bar\cC_x$) that $S$ intersects every gauge orbit in
$\Phi_+^{-1}(\mu)$ in a \emph{single} point. Indeed, it is instructive to verify that
if $\hat p$ is the special vertex of $\bar\cC_x$ for which
$\hat p_k=(1-k)|x|/2$ for $k=1,\dots,n$, then all points $K(\hat p,e^{\ri \hat q})$
lie on a single gauge orbit. This, and further inspection, can lead to the idea that
the variables $\hat q_j$ should be identified with arguments of complex numbers,
which lose their meaning at the origin that should correspond to the boundary of
$\bar\cC_x$. Our Theorem \ref{thm:4.9} will show that this idea is correct. It is proper to
stress that we arrived at such idea under the supporting influence of previous works
\cite{RIMS95,FK1}.
\end{remark}

\section{Characterization of the reduced system}
\label{sec:4}

The smoothness of the reduced phase space and the completeness of the reduced free
flows follows immediately if we can show that the gauge group $G_\mu$ acts in such a
way on $\Phi_+^{-1}(\mu)$ that the isotropy group of every point is just the finite
center of the symmetry group.
In Subsection \ref{subsec:4.1}, we prove that the factor of $G_\mu$ by the center
acts freely on $\Phi_+^{-1}(\mu)$. Then in Subsection \ref{subsec:4.2} we explain
that $\cC_x\times\T_n$ provides a model of a dense open subset of the reduced phase
space by means of the corresponding subset of $\Phi_+^{-1}(\mu)$ defined by
Proposition \ref{prop:3.3}. Adopting a key calculation from \cite{M}, it turns out that
$(\hat p,e^{\ri\hat q})\in\cC_x\times\T_n$ are Darboux coordinates on this dense open
subset. In Subsection \ref{subsec:4.3}, we demonstrate that the reduction
of the Abelian Poisson algebra of free Hamiltonians \eqref{T8} yields an integrable
system. Finally, in Subsection \ref{subsec:4.4}, we present a model of the full reduced phase
space, which is our main result.

\subsection{Smoothness of the reduced phase space}
\label{subsec:4.1}

It is clear that the normal subgroup of the full symmetry group $G_+\times G_+$
consisting of matrices of the form
\be
(\eta,\eta)\quad\text{with}\quad\eta=\diag(z\1_n,z\1_n),\quad z^{2n}=1
\label{S30}
\ee
acts trivially on the phase space. This subgroup is contained in $G_\mu$ \eqref{T25}.
The corresponding factor group of $G_\mu$ is called `effective gauge group' and is
denoted by $\bar G_\mu$. We wish to show that $\bar G_\mu$ acts freely on the
constraint surface $\Phi_+^{-1}(\mu)$.

We need the following elementary lemmas.

\begin{lemma}\label{lem:4.1}
Suppose that
\be
g_+\begin{bmatrix}\cos q&\ri\sin q\\\ri\sin q&\cos q\end{bmatrix}h_+
=g_+'\begin{bmatrix}\cos q&\ri\sin q\\\ri\sin q&\cos q\end{bmatrix}h_+'
\label{S31}
\ee
with $g_+,h_+,g_+',h_+'\in G_+$ and $q=\diag(q_1,\dots,q_n)$ subject to
\be
\frac{\pi}{2}\geq q_1>\dots>q_n>0.
\label{S32}
\ee
Then there exist diagonal matrices $m_1,m_2\in\T_n$ having the form
\be
m_1=\diag(a,\xi),\quad m_2=\diag(b,\xi),
\quad\xi\in\T_{n-1},\ a,b\in\T_1,\quad\det(m_1m_2)=1,
\label{*}
\ee
for which
\be
(g_+',h_+')=(g_+\diag(m_1,m_2),\diag(m_2^{-1},m_1^{-1})h_+).
\label{S33}
\ee
If \eqref{S32} holds with strict inequality $\frac{\pi}{2}>q_1$, then $m_1=m_2$,
i.e., $a=b$.
\end{lemma}

\begin{lemma}\label{lem:4.2}
Pick any $\hat p\in\bar\cC_x$ and consider the matrix
$\theta(x,\hat p)$ given by \eqref{S19} and \eqref{S20}. Then the entries
$\theta_{n,1}(x,\hat p)$ and $\theta_{j,j+1}(x,\hat p)$ are all non-zero if $x>0$ and
the entries $\theta_{1,n}(x,\hat p)$ and $\theta_{j+1,j}(x,\hat p)$ are all non-zero
if $x<0$.
\end{lemma}

For convenience, we present the proof of Lemma \ref{lem:4.1} in Appendix \ref{sec:C}.
The property recorded in Lemma \ref{lem:4.2} is known \cite{RIMS95,FK1}, and is
easily checked by inspection.

\begin{proposition}\label{prop:4.3}
The effective gauge group $\bar G_\mu$ acts freely on $\Phi_+^{-1}(\mu)$.
\end{proposition}

\begin{proof}
Since every gauge orbit intersects the set $S$ specified by Proposition
\ref{prop:3.3}, it is enough to show that if $(\eta_L,\eta_R)\in G_\mu$ maps $K\in S$
\eqref{S27} to itself, then $(\eta_L,\eta_R)$ equals some element $(\eta,\eta)$ given
in \eqref{S30}. For $K$ of the form \eqref{S3}, we can spell out
$K'\equiv\eta_L K\eta_R^{-1}$ as
\be
K'=\begin{bmatrix}\eta_L(1)\rho&\0_n\\\0_n&\eta_L(2)\end{bmatrix}
\begin{bmatrix}\cos q&\ri\sin q\\\ri\sin q&\cos q\end{bmatrix}
\begin{bmatrix}\eta_R(1)^{-1}&\0_n\\\0_n&\eta_R(2)^{-1}\end{bmatrix}
\begin{bmatrix}e^{-v}\1_n&\eta_R(1)\alpha\eta_R(2)^{-1}\\\0_n&e^v\1_n\end{bmatrix}.
\label{S34}
\ee
The equality $K'=K$ implies by the uniqueness of the Iwasawa decomposition and Lemma
\ref{lem:4.1} that we must have
\be
\eta_L(2)=\eta_R(1)=m_2,\quad\eta_R(2)=m_1,\quad\eta_L(1)\rho=\rho m_1,
\label{S35}
\ee
with some diagonal unitary matrices having the form \eqref{*}. By using that
$\eta_R(1)=m_2$ and $\eta_R(2)=m_1$, the Iwasawa decomposition of $K'=K$ in
\eqref{S27} also entails the relation
\be
\alpha=m_2\alpha m_1^{-1}.
\label{S36}
\ee
Because of \eqref{S28}, the off-diagonal components of the matrix equation
\eqref{S36} yield
\be
\theta(-x,\hat p)_{jk}=\big(m_2\theta(-x,\hat p)m_1^{-1}\big)_{jk},
\quad\forall j\neq k.
\ee
This implies by means of Lemma \ref{lem:4.2} and equation \eqref{*} that $m_1=m_2=z\1_n$ is a
scalar matrix. But then $\eta_L(1)=m_1$ follows from $\eta_L(1)\rho=\rho m_1$, and
the proof is complete.
\end{proof}

\noindent
Proposition \ref{prop:4.3} and the general results gathered in Appendix \ref{sec:D}
imply the following theorem, which is one of our main results.

\begin{theorem}\label{thm:4.4}
The constraint surface $\Phi_+^{-1}(\mu)$ is an embedded submanifold of $\SL(2n,\C)$ and
the reduced phase space $M$ \eqref{T24} is a smooth manifold for which the natural
projection $\pi_\mu\colon\Phi_+^{-1}(\mu)\to M$ is a smooth submersion.
\end{theorem}

\subsection{Model of a dense open subset of the reduced phase space}
\label{subsec:4.2}

Let us denote by $S^o\subset S$ the subset of the elements $K$ given by Proposition
\ref{prop:3.3} with $\hat p$ in the interior $\cC_x$ of the polyhedron $\bar\cC_x$ \eqref{S18}.
Explicitly, we have
\be
S^o=\{K(\hat p,e^{\ri\hat q})\mid(\hat p,e^{\ri\hat q})\in\cC_x\times\T_n\},
\label{S37}
\ee
where $K(\hat p,e^{\ri\hat q})$ stands for the expression \eqref{S27}. Note that $S^o$
is in bijection with $\cC_x\times\T_n$. The next lemma says that no two different
point of $S^o$ are gauge equivalent.

\begin{lemma}\label{lem:4.5}
The intersection of any gauge orbit with $S^o$ consists of at most one point.
\end{lemma}

\begin{proof}
Suppose that
\be
K':=K(\hat p',e^{\ri\hat q'})=\eta_L K(\hat p,e^{\ri\hat q})\eta_R^{-1}
\label{S38}
\ee
with some $(\eta_L,\eta_R)\in G_\mu$. By spelling out the gauge transformation as in
\eqref{S34}, using the shorthand $\sin q=e^{\hat p}$, we observe that $\hat p'=\hat p$
since $q$ in \eqref{S1} does not change under the action of $G_+\times G_+$. Since
now we have $\frac{\pi}{2}>q_1$ (which is equivalent to $0>\hat p_1$), the arguments
applied in the proof of Proposition \ref{prop:4.3} permit to translate the equality \eqref{S38}
into the relations
\be
\eta_L(2)=\eta_R(1)=\eta_R(2)=m,\quad\eta_L(1)\rho=\rho m,
\label{**}
\ee
complemented with the condition
\be
\alpha(\hat p,e^{\ri\hat q'})=m\alpha(\hat p,e^{\ri\hat q})m^{-1},
\label{S39}
\ee
which is equivalent to
\be
e^{\ri\hat q'}\theta(-x,\hat p)=me^{\ri\hat q}\theta(-x,\hat p)m^{-1}.
\label{S40}
\ee
We stress that $m\in\T_n$ and notice from \eqref{S20} that for $\hat p\in\cC_x$
all the diagonal entries $\theta(-x,\hat p)_{jj}$ are non-zero. Therefore we conclude
from \eqref{S40} that $e^{\ri\hat q'}=e^{\ri q}$. This finishes the proof, but of
course we can also confirm that $m=z\1_n$, consistently with Proposition \ref{prop:4.3}.
\end{proof}

Now we introduce the map $\cP\colon\SL(2n,\C)\to\R^n$ by
\be
\cP\colon K=g_Lb_R^{-1}\mapsto\hat p,
\label{S42}
\ee
defined by writing $g_L$ in the form \eqref{S1} with $\sin q=e^{\hat p}$.
The map $\cP$ gives rise to a map $\bar \cP\colon M\to\R^n$ verifying
\be
\bar \cP(\pi_\mu(K))=\cP(K),\quad\forall K\in\Phi_+^{-1}(\mu),
\label{43}
\ee
where $\pi_\mu$ is the canonical projection \eqref{T28}. We notice that, since the
`eigenvalue parameters' $\hat p_j$ $(j=1,\dots,n)$ are pairwise different for any
$K\in\Phi_+^{-1}(\mu)$, $\bar \cP$ is a smooth map. The continuity of $\bar \cP$ implies
that
\be
M^o:=\bar \cP^{-1}(\cC_x)=\pi_\mu(S^o)\subset M
\label{S45}
\ee
is an open subset. The second equality is a direct consequence of our foregoing
results about $S$ and $S^o$. Note that $\bar \cP^{-1}(\bar\cC_x)=\pi_\mu(S)=M$.
Since $\pi_\mu$ is continuous (actually smooth) and any point of $S$ is the limit of
a sequence in $S^o$, $M^o$ is \emph{dense} in the reduced phase space $M$.
The dense open subset $M^o$ can be parametrized by $\cC_x\times\T_n$ according to
\be
(\hat p,e^{\ri\hat q})\mapsto\pi_\mu(K(\hat p,e^{\ri\hat q})),
\label{S46}
\ee
which also allows us to view $S^o\simeq\cC_x\times\T_n$ as a model of $M^o\subset M$.
In principle, the restriction of the reduced symplectic form to $M^o$ can now be
computed by inserting the explicit formula $K(\hat p,e^{\ri\hat q})$ \eqref{S27}
into the Alekseev-Malkin form \eqref{T5}. In the analogous reduction of the
Heisenberg double of $\SU(n,n)$, Marshall \cite{M} found a nice way to circumvent such
a tedious calculation. By taking the same route, we have verified that $\hat p$ and
$\hat q$ are Darboux coordinates on $M^o$.

The outcome of the above considerations is summarized by the next theorem.

\begin{theorem}\label{thm:4.6}
$M^o$ defined by equation \eqref{S45} is a dense open subset of
the reduced phase space $M$. Parametrizing $M^o$ by $\cC_x\times\T_n$ according to
\eqref{S46}, the restriction of reduced symplectic form $\omega_M$ \eqref{T29} to
$M^o$ is equal to $\hat\omega=\sum_{j=1}^nd\hat q_j\wedge d\hat p_j$ \eqref{I4}.
\end{theorem}

\subsection{Liouville integrability of the reduced free Hamiltonians}
\label{subsec:4.3}

The Abelian Poisson algebra $\fH$ \eqref{T8} consists of $(G_+ \times G_+)$-invariant
functions\footnote{More precisely, $\fH=C^\infty(\SL(2n,\C))^{\SU(2n)\times\SU(2n)}$.}
generating complete flows, given explicitly by \eqref{T10}, on the unreduced phase
space. Thus each element of $\fH$ descends to a smooth reduced Hamiltonian on $M$
\eqref{T24}, and generates a complete flow via the reduced symplectic form
$\omega_M$. This flow is the projection of the corresponding unreduced flow,
which preserves the constraint surface $\Phi_+^{-1}(\mu)$. It also follows from the
construction that $\fH$ gives rise to an Abelian Poisson algebra, $\fH_M$, on
$(M,\omega_M)$. Now the question is whether the Hamiltonian vector fields of $\fH_M$
span an $n$-dimensional subspace of the tangent space at the points of a dense open
submanifold of $M$. If yes, then $\fH_M$ yields a Liouville integrable system, since
$\dim(M)=2n$.

Before settling the above question, let us focus on the Hamiltonian $\cH\in \fH$
defined by
\be
\cH(K):=\frac{1}{2}\tr\big((K^\dag K)^{-1}\big)=\frac{1}{2}\tr(b_R^\dag b_R).
\label{F19}
\ee
Using the formula of $K(\hat p,e^{\ri\hat q})$ in Remark \ref{rem:3.4}, it is readily
verified that
\be
\cH(K(\hat p,e^{\ri\hat q}))=H(\hat p,\hat q;x,u,v),
\quad\forall(\hat p,e^{\ri\hat q})\in\cC_x\times\T_n,
\label{F20}
\ee
with the Hamiltonian $H$ displayed in equation \eqref{I1}. \emph{Consequently, $H$ in
\eqref{I1} is identified as the restriction of the reduction of $\cH$ \eqref{F19}
to the dense open submanifold $M^o$ \eqref{S45} of the reduced phase space, wherein
the flow of every element of $\fH_M$ is complete.}

Turning to the demonstration of Liouville integrability, consider the $n$ functions
\be
\cH_k(K):=\frac{1}{2k}\tr\big((K^\dag K)^{-1}\big)^k
=\frac{1}{2k}\tr(b_R^\dag b_R)^k,\quad k=1,\dots,n.
\label{F21}
\ee
The restriction of the corresponding elements of $\fH_M$ on $M^o\simeq\cC_x\times\T_n$
gives the functions
\be
H_k(\hat p,\hat q)=\frac{1}{2k}\tr
\begin{bmatrix}
e^{2v}\1_n&-e^v\alpha\\
-e^v\alpha^\dag&(e^{-2v}\1_n+\alpha^\dag\alpha)\end{bmatrix}^k,
\label{F22}
\ee
where $\alpha$ has the form \eqref{S28}. These are real-analytic functions on
$\cC_x\times\T_n$. It is enough to show that their exterior derivatives are
linearly independent on a dense open subset of $\cC_x\times\T_n$. This follows if we
show that the function
\be
f(\hat p,\hat q)=\det\big[d_{\hat q}H_1,d_{\hat q}H_2,\dots,d_{\hat q}H_n\big]
\label{F23}
\ee
is not identically zero on $\cC_x\times\T_n$. Indeed, since $f$ is an analytic
function and $\cC_x\times\T_n$ is connected, if $f$ is not identically zero then
its zero set cannot contain any accumulation point. This, in turn, implies that $f$ is
non-zero on a dense open subset of $\cC_x\times\T_n\simeq M^o$, which is also dense
and open in the full reduced phase space $M$. In other words, the reductions of
$\cH_k$ $(k=1,\ldots, n)$ possess the property of Liouville integrability.
It is rather obvious that the function $f$ is not identically zero, since $H_k$
involves dependence on $\hat q$ through $e^{\pm\ri k\hat q}$ and lower powers of
$e^{\pm\ri\hat q}$. It is not difficult to inspect the function $f(\hat p,\hat q)$
in the `asymptotic domain' where all
differences $|\hat p_j-\hat p_m|$ $(m\neq j)$ tend to infinity, since in this domain
$\alpha$ becomes close to a diagonal matrix. We omit the details
of this inspection, whereby we checked that $f$ is indeed not identically zero.

The above arguments prove the Liouville integrability of the reduced free
Hamiltonians, i.e., the elements of $\fH_M$. Presumably, there exists a dual set of
integrable many-body Hamiltonians that live on the space of action-angle variables of
the Hamiltonians in $\fH_M$. The construction of such dual Hamiltonians is an
interesting task for the future, which will be further commented upon in Section
\ref{sec:5}.

\subsection{The global structure of the reduced phase space}
\label{subsec:4.4}

We here construct a global cross-section of the gauge orbits in the constraint
surface $\Phi_+^{-1}(\mu)$. This engenders a symplectic diffeomorphism between
the reduced phase space $(M,\omega_M)$ and the manifold $(\hat M_c,\hat\omega_c)$
below. It is worth noting that $(\hat M_c, \hat \omega_c)$ is symplectomorphic
to $\R^{2n}$ carrying the standard Darboux 2-form, and one can easily find an
explicit symplectomorphism if desired. Our construction was inspired by the
previous papers \cite{RIMS95,FK1}, but detailed inspection of the specific
example was also required for finding the final result given by Theorem \ref{thm:4.9}.
After a cursory glance,
the reader is advised to go directly to this theorem and follow the definitions
backwards as becomes necessary. See also Remark \ref{rem:4.10} for the rationale behind the
subsequent definitions.

To begin, consider the product manifold
\be
\hat M_c:=\C^{n-1}\times\D,
\label{hatM}
\ee
where $\D$ stands for the
open unit disk, i.e., $\D:=\{w\in\C:|w|<1\}$, and equip it with the symplectic form
\be
\hat\omega_c=\ri\sum_{j=1}^{n-1}dz_j\wedge d\bar z_j
+\frac{\ri dz_n\wedge d\bar z_n}{1-z_n\bar z_n}.
\label{hatom}
\ee
The subscript $c$ refers to `complex variables'. Define the surjective map
\be
\hat\cZ_x\colon\bar\cC_x\times\T_n\to\hat M_c,\quad
(\hat p,e^{\ri\hat q})\mapsto z(\hat p,e^{\ri\hat q})
\label{G24}
\ee
by the formulae
\be
\begin{gathered}
z_j(\hat p,e^{\ri\hat q})
=(\hat p_j-\hat p_{j+1}-|x|/2)^{\tfrac{1}{2}}\prod_{k=j+1}^ne^{\ri\hat q_k},
\quad j=1,\dots,n-1,\\
z_n(\hat p,e^{\ri\hat q})=(1-e^{\hat p_1})^{\tfrac{1}{2}}\prod_{k=1}^ne^{\ri\hat q_k}.
\end{gathered}
\label{G25}
\ee
Notice that the restriction $\cZ_x$ of $\hat \cZ_x$ to $\cC_x\times \T_n$ is a
diffeomorphism onto the dense open submanifold
\be
\hat M_c^o=\{z\in\hat M_c\mid\prod_{j=1}^nz_j\neq 0\}.
\ee
It verifies
\be
\cZ_x^\ast(\hat\omega_c)=\hat\omega=\sum_{j=1}^nd\hat q_j\wedge d\hat p_j,
\label{densehatom}
\ee
which means that $\cZ_x$ is a symplectic embedding of $(\cC_x\times\T_n,\hat\omega)$
into $(\hat M_c,\hat\omega_c)$. The inverse
$\cZ_x^{-1}\colon\hat M_c^o\to\cC_x\times\T_n$ operates according to
\be
\begin{gathered}
\hat p_1(z)=\log(1-|z_n|^2),\quad
\hat p_j(z)=\log(1-|z_n|^2)-\sum_{k=1}^{j-1}(|z_k|^2+|x|/2)\quad (j=2,\dots,n)\\
e^{\ri\hat q_1}(z)=\frac{z_n\bar z_1}{|z_n\bar z_1|},\quad
e^{\ri\hat q_m}(z)=\frac{z_{m-1}\bar z_m}{|z_{m-1}\bar z_m|}\quad (m=2,\dots,n-1),\quad
e^{\ri\hat q_n}(z)=\frac{z_{n-1}}{|z_{n-1}|}.
\end{gathered}
\label{G26}
\ee
It is important to remark that the $\hat p_k(z)$ $(k=1,\dots,n)$ given above yield
smooth functions on the whole of $\hat M_c$, while the angles $\hat q_k$ are of
course not well-defined on the complementary locus of $\hat M_c^o$.
Our construction of the global cross-section will rely on the building blocks
collected in the following long definition.

\begin{definition}\label{def:4.7}
For any $(z_1,\dots,z_{n-1})\in\C^{n-1}$ consider the smooth
functions
\be
Q_{jk}(x,z)=\bigg[\frac{\sinh(\sum_{\ell=j}^{k-1}z_\ell\bar z_\ell+(k-j)|x|/2-x/2)}
{\sinh(\sum_{\ell=j}^{k-1}z_\ell\bar z_\ell+(k-j)|x|/2)}\bigg]^{\tfrac{1}{2}},\quad
1\leq j<k\leq n,
\label{G27}
\ee
and set $Q_{jk}(x,z):=Q_{kj}(-x,z)$ for $j>k$. Applying these as well as the real analytic function
\be
J(y):=\sqrt{\frac{\sinh(y)}{y}},\quad y\neq 0,\quad J(0):=1,
\label{G28}
\ee
and recalling \eqref{S21}, introduce the $n\times n$ matrix $\hat\zeta(x,z)$ by the
formulae
\be
\begin{gathered}
\hat\zeta(x,z)_{aa}=r(x,\hat p(z))_a,\quad
\hat\zeta(x,z)_{aj}=-\overline{\hat\zeta(x,z)_{ja}},\quad j\neq a, \\
\hat\zeta(x,z)_{jn}=\sqrt{\frac{\sinh(\frac{x}{2})}{\sinh(\frac{nx}{2})}}
\frac{z_jJ(z_j\bar z_j)}{\sinh(z_j\bar z_j+\frac{x}{2})}
\prod_{\substack{\ell=1\\(\ell\neq j,j+1)}}^nQ_{j\ell}(x,z),\quad x>0,\quad j\neq n,\\
\hat\zeta(x,z)_{j1}=\sqrt{\frac{\sinh(\frac{x}{2})}{\sinh(\frac{nx}{2})}}
\frac{\bar z_{j-1}J(z_{j-1}\bar z_{j-1})}{\sinh(z_{j-1}\bar z_{j-1}-\frac{x}{2})}
\prod_{\substack{\ell=1\\(\ell\neq j-1,j)}}^nQ_{j\ell}(x,z),\quad x<0,\quad j\neq 1,\\
\hat\zeta(x,z)_{jk}=\delta_{j,k}
+\frac{\hat\zeta(x,z)_{ja}\hat\zeta(x,z)_{ak}}{1+\hat\zeta(x,z)_{aa}},\quad j,k\neq a,
\end{gathered}
\label{G29}
\ee
where $a=n$ if $x>0$ and $a=1$ if $x<0$. Then introduce the matrix $\hat\theta(x,z)$
for $x>0$ as
\be
\begin{gathered}
\hat\theta(x,z)_{jk}=\frac{\sinh(\frac{nx}{2})\sgn(k-j-1)
\hat\zeta(x,z)_{jn}\hat\zeta(-x,\bar z)_{1k}}
{\sinh(\sum_{\ell=\min(j,k)}^{\max(j,k)-1}z_\ell\bar z_\ell+|k-j-1|\frac{x}{2})},\quad
k\neq j+1,\\
\hat\theta(x,z)_{j,j+1}=\frac{-\sinh(\frac{x}{2})}{\sinh(z_j\bar z_j+\frac{x}{2})}
\prod_{\substack{\ell=1\\(\ell\neq j,j+1)}}^nQ_{j\ell}(x,z)Q_{j+1,\ell}(-x,z),
\end{gathered}
\label{G30}
\ee
and for $x<0$ as
\be
\hat\theta(x,z)=\hat\theta(-x,\bar z)^\dag.
\label{G31}
\ee
Finally, for any $z\in\hat M_c$ define the matrix
$\hat\gamma(x,z)=\diag(\hat\gamma_1,\dots,\hat\gamma_n)$ with
\be
\hat\gamma(z)_1=z_n\sqrt{2-z_n\bar z_n},\quad
\hat\gamma(x,z)_j=\bigg[1-(1-z_n\bar z_n)^2
e^{-2\sum_{\ell=1}^{j-1}(z_\ell\bar z_\ell+|x|/2)}\bigg]^{\frac{1}{2}},\quad j=2,\dots,n,
\label{G32}
\ee
and the matrix
\be
\hat\alpha(x,u,v,z)=-\ri\big[\sqrt{e^{-2u}e^{-2\hat p(z)}
-e^{-2v}\1_n}]\,\hat\theta(-x,\bar z)
-e^ve^{-\hat p(z)}\hat\gamma(x,z)^\dag\big],
\label{hatal}
\ee
using the constants $x,u,v$ subject to \eqref{I2}.
\end{definition}

Although the variable $z_n$ appears only in $\hat\gamma_1$, we can regard all
objects defined above as smooth functions on $\hat M_c$, and  we shall do so below.

The key properties of the matrices $\hat\zeta$, $\hat\theta$, $\hat\alpha$ and
$\hat\gamma$ are given by the following lemma, which can be verified by
straightforward inspection. The role of these identities and their origin will be
enlightened by Theorem \ref{thm:4.9}.

\begin{lemma}\label{lem:4.8}
Prepare the notations
\be
\tau_{(x)}:=\diag(\tau_2,\dots,\tau_n,1)\quad \text{if}\quad x>0
\quad\text{and}\quad
\tau_{(x)}:=\diag(1,\tau_2^{-1},\dots,\tau_n^{-1})\quad\text{if}\quad x<0,
\label{G34}
\ee
\be
\tilde\tau_{(x)}:=\diag(1,\tau_2,\dots,\tau_n)\quad \text{if}\quad x>0
\quad\text{and}\quad
\tilde\tau_{(x)}:=\diag(\tau_2^{-1},\dots,\tau_n^{-1},1)\quad \text{if}\quad x<0
\label{G35}
\ee
with
\be
\tau_j=\prod_{k=j}^ne^{\ri\hat q_k}.
\ee
Then the following identities hold for all
$(\hat p,e^{\ri\hat q})\in\bar\cC_x\times\T_n$:
\begin{align}
\hat\zeta(x,z(\hat p,e^{\ri \hat q}))
&=\tau_{(x)}\zeta(x,\hat p)\tau_{(x)}^{-1},\\
\hat\theta(x,z(\hat p,e^{\ri \hat q}))
&=\tau_{(x)}\theta(x,\hat p)\tilde\tau_{(x)}^{-1},\\
\hat\gamma(x,z(\hat p,e^{\ri \hat q}))
&=e^{\ri\hat q}\tau_{(x)}\tilde\tau_{(x)}^{-1}\sqrt{\1_n-e^{2\hat p}},\\
\hat\alpha(x,u,v, z(\hat p,e^{\ri \hat q}))
&=e^{-\ri \hat q} \tilde\tau_{(x)}
\alpha(x,u,v,\hat p, e^{\ri \hat q})\tau_{(x)}^{-1}.
\label{}
\end{align}
Here we use Definition \ref{def:4.7} and the functions on $\bar\cC_x\times\T_n$
introduced in Subsection \ref{subsec:3.2}.
\end{lemma}

For the verification of the above identities, we remark that the vector $r$
\eqref{S21} can be expressed as a smooth function of the complex variables as
\be
r(x,\hat p(z))_j=\sqrt{\frac{\sinh(\frac{x}{2})}{\sinh(\frac{nx}{2})}}
\prod_{\substack{k=1\\(k\neq j)}}^nQ_{jk}(x,z),\quad j=1,\dots,n.
\label{G33}
\ee

With all necessary preparations now done, we state the main new result of the paper.

\begin{theorem}\label{thm:4.9}
The image of the smooth map
$\hat K\colon\hat M_c\to\SL(2n,\C)$ given by the formula
\be
\hat K(z)=\begin{bmatrix}\kappa(x)\hat\zeta(x,z)^{-1} &\0_n\\\0_n&\1_n\end{bmatrix}
\begin{bmatrix}\hat\gamma(x,z)&\ri e^{\hat p(z)}\\\ri e^{\hat p(z)}&
\hat\gamma(x,z)^\dag\end{bmatrix}
\begin{bmatrix}e^{-v}\1_n&\hat\alpha(x,u,v,z)\\\0_n&e^v\1_n\end{bmatrix}
\label{G46}
\ee
lies in $\Phi_+^{-1}(\mu)$, intersects every gauge orbit in precisely one point, and $\hat K$ is injective.
The pull-back of the Alekseev-Malkin 2-form $\omega$ \eqref{T5} by $\hat K$ is
$\hat\omega_c$ \eqref{hatom}. Consequently, $\pi_\mu\circ\hat K\colon\hat M_c\to M$
is a symplectomorphism, whereby $(\hat M_c,\hat\omega_c)$ provides a model of the
reduced phase space $(M,\omega_M)$ defined in Subsection \ref{subsec:2.2}.
\end{theorem}

\begin{proof}
The proof is based upon the identity
\be
\hat K(z(\hat p,e^{\ri\hat q}))
=\begin{bmatrix}\kappa(x)\tau_{(x)}\kappa(x)^{-1}&\0_n\\
\0_n&\tilde\tau_{(x)}e^{-\ri\hat q}\end{bmatrix}K(\hat p,e^{\ri\hat q})
\begin{bmatrix}\tilde\tau_{(x)}e^{-\ri\hat q}&\0_n\\
\0_n&\tau_{(x)}\end{bmatrix}^{-1}, \quad \forall (\hat p, e^{\ri \hat q}) \in \bar\cC_x \times \T_n,
\label{G47}
\ee
which is readily seen to be equivalent to the set of identities displayed in Lemma
\ref{lem:4.8}. It means that $\hat K(z(\hat p,e^{\ri\hat q}))$ is a gauge transform of
$K(\hat p, e^{\ri \hat q})$ in \eqref{S27}. Indeed, the above transformation of
$K(\hat p,e^{\ri\hat q})$ has the form \eqref{T20} with
\be
\eta_L=c\begin{bmatrix}\kappa(x)\tau_{(x)}\kappa(x)^{-1}&\0_n\\
\0_n&\tilde\tau_{(x)}e^{-\ri\hat q}\end{bmatrix},
\qquad
\eta_R=c\begin{bmatrix}\tilde\tau_{(x)}e^{-\ri\hat q}&\0_n\\
\0_n&\tau_{(x)}\end{bmatrix},
\label{G48}
\ee
where $c$ is a harmless scalar inserted to ensure $\det(\eta_L)=\det(\eta_R)=1$.
Using \eqref{S25} and \eqref{G34}, one can check that $\kappa(x)\tau_{(x)}\kappa(x)^{-1}\hat v(x)=\hat v(x)$
for the vector $\hat v(x)$ in \eqref{S24}, which implies via the relation \eqref{S15}
that $(\eta_L,\eta_R)$ belongs to the isotropy group $G_\mu$ \eqref{T25},
the gauge group acting on $\Phi_+^{-1}(\mu)$.

It follows from Proposition \ref{prop:3.3} and the identity \eqref{G47} that the set
\be
\hat S:=\{\hat K(z)\mid z\in\hat M_c\}
\label{G49}
\ee
lies in $\Phi_+^{-1}(\mu)$ and intersects every gauge orbit. Since the dense subset
\be
\hat S^o:=\{\hat K(z)\mid z\in\hat M_c^o\}
\label{G50}
\ee
is gauge equivalent to $S^o$ in \eqref{S37}, we obtain the equality
\be
\hat K^\ast(\omega)=\hat\omega_c
\label{G51}
\ee
by using Theorem \ref{thm:4.6} and equation \eqref{densehatom}. More precisely, we here also
utilized that $\hat K^\ast(\omega)$ is (obviously) smooth and $\hat M_c^o$ is dense
in $\hat M_c$.

The only statements that remain to be proved are that the intersection of $\hat S$
with any gauge orbit consists of a single point and that $\hat K$ is injective. (These are already clear
for $\hat S^o\subset \hat S$ and for $\hat K\vert_{\hat M^o_c}$.) Now suppose that
\be
\hat K(z')=\begin{bmatrix}\eta_L(1)&\0_n\\\0_n&\eta_L(2)\end{bmatrix}
\hat K(z)\begin{bmatrix}\eta_R(1) &\0_n\\\0_n&\eta_R(2)\end{bmatrix}^{-1}
\label{G52}
\ee
for some gauge transformation and $z,z'\in\hat M_c$.
Let us observe from the definitions that we can write
\be
\begin{bmatrix}\hat\gamma(x,z)&\ri e^{\hat p(z)}\\\ri e^{\hat p(z)}&
\hat\gamma(x,z)^\dag\end{bmatrix}
=D(z)\begin{bmatrix}\cos q(z)&\ri\sin q(z)\\\ri\sin q(z)&\cos q(z)\end{bmatrix}D(z),
\label{G53}
\ee
where $\sin q(z)=e^{\hat p(z)}$, with $\frac{\pi}{2}\geq q_1>\dots>q_n>0$,
and $D(z)$ is a diagonal unitary matrix of the form
$D(z)=\diag(d_1,\1_{n-1},\bar d_1,\1_{n-1})$.
Then the uniqueness properties of the Iwasawa decomposition of $\SL(2n,\C)$ and the
generalized Cartan decomposition \eqref{S1} of $\SU(2n)$ allow to establish the
following consequences of \eqref{G52}. First,
\be
\hat p(z)=\hat p(z').
\label{G54}
\ee
Second, using Lemma \ref{lem:4.1},
\be
\begin{bmatrix}\eta_R(1)&\0_n\\\0_n&\eta_R(2)\end{bmatrix}
=\begin{bmatrix}m_2&\0_n\\\0_n&m_1\end{bmatrix}
\label{G55}
\ee
for some diagonal unitary matrices of the form \eqref{*}. Third, we have
\be
\hat\alpha(z')=\eta_R(1)\hat\alpha(z)\eta_R(2)^{-1}=m_2\hat\alpha(z)m_1^{-1}.
\label{G56}
\ee
For definiteness, let us focus on the case $x>0$. Then we see from the definitions
that the components $\hat \alpha_{k+1,k}$ and $\hat\alpha_{1,n}$ depend only on
$\hat p(z)$ and are non-zero. By using this, we find from \eqref{G56} that
$m_1=m_2=C\1_n$ with a scalar $C$, and therefore
\be
\hat\alpha(z')=\hat\alpha(z).
\label{G57}
\ee
Inspection of the components $(1,2),\dots,(1,n-1)$ of this matrix equality and
\eqref{G54} permit to conclude that $z'_2=z_2,\dots,z'_{n-1}=z_{n-1}$, respectively.
Then, the equality of the $(2,n)$ entries in \eqref{G57} gives $z'_1=z_1$ which
used in the $(1,1)$ position implies $z'_n=z_n$. Thus we see that $z'=z$ and
the proof is complete. (Everything written below \eqref{G56} is quite similar for
$x<0$.)
\end{proof}

\begin{remark}\label{rem:4.10}
Let us hint at the way the global structure was found. The first point to notice was
that all or some of the phases $e^{\ri \hat q_j}$ cannot encode gauge invariant
quantities if $\hat p$ belongs to the boundary of $\bar \cC_x$, as was already
mentioned in Remark \ref{rem:3.5}. Motivated by \cite{FK1}, then we searched for
complex variables by requiring that a suitable gauge transform of
$K(\hat p,e^{\ri \hat q})$ in \eqref{S27} should be expressible as a smooth function
of those variables. Given the similarities to \cite{FK1}, only the definition of
$z_n$ was a true open question. After trial and error, the idea came in a flash that
the gauge transformation at issue should be constructed from a transformation that
appears in Lemma \ref{lem:C.1}. Then it was not difficult to find the correct result.
\end{remark}

\begin{remark}\label{rem:4.11}
Let us elaborate on how the trajectories $\hat p(t)$ corresponding to the flows of the reduced
free Hamiltonians, arising from $\cH_k$ \eqref{F21} for $k=1,\ldots, n$, can be obtained.
Recall that for $k=1$ the reduction of $\cH_1$ completes the main Hamiltonian $H$ \eqref{I1}.
Since $\cH_k(K)= h_k(b_R)$ with $h_k(b) = \frac{1}{2k}\tr (b^\dag b)^k$, the free flow
generated by $\cH_k$ through the initial value $K(0) = g_L(0) b_R^{-1}(0)$ is given by
\eqref{T10} with $d^R h_k(b) = \ri [(b^\dagger b)^k- \frac{1}{2n}\tr(b^\dagger b)^k \1_{2n}]$.
Thus  the curve $g_L(t)$ \eqref{T10} has the form
\begin{equation}
g_L(t)=g_L(0)\exp\!\big(-\ri t\big[\cL(0)^k- \frac{1}{2n}\tr (\cL(0)^k) \1_{2n}\big]\big)
\quad\text{with}\quad
\cL(0)= b_R(0)^\dag b_R(0).
\label{}
\end{equation}
The reduced flow results by the usual projection algorithm.
This starts by picking an initial value $z(0) \in \hat M_c$ and setting
$K(0)=\hat K(z(0))$
by applying \eqref{G46}, which directly determines $g_L(0)$ and $b_R(0)$ as well.
Then the map $\cP$ \eqref{S42} gives rise to $\hat p(t)$ via the decomposition of
$g_L(t)\in\SU(2n)$ as displayed in \eqref{S1}, that is
\begin{equation}
\hat p(t)=\cP(K(t)).
\label{}
\end{equation}
More explicitly, if $\cD(t)$ stands for the (11) block of $g_L(t)$, then the eigenvalues
of $\cD(t)\cD(t)^\dag$  are
\begin{equation}
\sigma(\cD(t)\cD(t)^\dag)=\{\cos^2q_j(t)\mid j=1,\dots,n\},
\label{}
\end{equation}
from which $\hat p_j(t)$ can be obtained using \eqref{S13}.
In particular,  the `particle positions' evolve according to an `eigenvalue dynamics'  similarly to other
many-body systems. This involves the one-parameter group
$\exp\!\big(-\ri t  \big[\cL(0)^k- \frac{1}{2n}\tr (\cL(0)^k) \1_{2n} \big] \big)$, where $ \cL(0)$
is the initial value of the Lax matrix (cf. \eqref{F22})
\begin{equation}
\cL(z)=\begin{bmatrix}
e^{2v}\1_n&-e^v\hat\alpha(z)\\
-e^v\hat\alpha(z)^\dag&(e^{-2v}\1_n+\hat\alpha(z)^\dag\hat\alpha(z))
\end{bmatrix},
\label{4.61}
\end{equation}
where we suppressed the dependence of $\hat \alpha$ \eqref{hatal} on the parameters $x,u,v$.
A more detailed characterization of the dynamics will be provided elsewhere.
\end{remark}

\section{Discussion and outlook on open problems}
\label{sec:5}

In this paper we derived a deformation of the trigonometric $\BC_n$ Sutherland system
by means of Hamiltonian reduction of a free system on the Heisenberg double of
$\SU(2n)$. Our main result is the global characterization of the reduced phase space
given by Theorem \ref{thm:4.9}. The Liouville integrability of our system holds on
this phase space, wherein the reduced free flows are complete. These flows can be
obtained by the usual projection method applied to the original free flows described
in Section \ref{sec:2}.

The local form of our reduced `main Hamiltonian' \eqref{I1} is similar to the
Hamiltonian derived in \cite{M}, which deforms the hyperbolic $\BC_n$ Sutherland
system. However, besides a sign difference corresponding to the difference of the
undeformed Hamiltonians, the local domain of our system, $\cC_x\times\T_n$ in
\eqref{I3}, is different from the local domain appearing in \cite{M}, which in effect
has the form $\cC_x'\times\T_n$ with the open polyhedron\footnote{The notational
correspondence between \cite{M} and the present paper is:
$(q,p,\alpha,x,y)\leftrightarrow(\hat p,\hat q,e^{-\frac{x}{2}},e^{-v},e^{-u})$.}
\be
\cC'_x:=\{\hat p\in\R^n\mid  \hat p_k-\hat p_{k+1}>|x|/2\
(k=1,\dots,n-1)\}.
\label{D1}
\ee
We here wish to point out that $\cC_x'\times\T_n$ is \emph{not} the full reduced
phase space that arises from the reduction considered in \cite{M}. In fact, similarly
to our case, the constraint surface contains a submanifold of the form
$\bar\cC_x'\times\T_n$ in the case of \cite{M}, where $\bar\cC_x'$ is the closure of
$\cC_x'$. Then a global model of the reduced phase space can be constructed by
introducing complex variables suitably accommodating the procedure that we utilized
in Subsection \ref{subsec:4.4}. Erroneously, in \cite{M} the full phase space was
claimed to be $\cC_x'\times\T_n$; the details of the correct description will be
presented elsewhere.

Throughout the text we assumed that $n>1$, but we now note that the reduced system
can be specialized to $n=1$ and the reduction procedure works in this case as well.
The assumption was made merely to save words. The formalism actually simplifies for
$n=1$ since the Poisson structure on $G_+=\mathrm{S}(\UN(1)\times\UN(1))<\SU(2)$ is
trivial.

As explained in Appendix \ref{sec:A}, the Hamiltonian \eqref{I1} is a singular limit
of a specialization of the trigonometric van Diejen Hamiltonian \cite{vD}, which (in
addition to the deformation parameter) contains 5 coupling constants. As a result,
at least classically, van Diejen's system can be degenerated into the trigonometric
$\BC_n$ Sutherland system either directly, as described in \cite{vD}, or in a
roundabout way, going through our system. Of course, a similar statement holds in
relation to hyperbolic $\BC_n$ Sutherland and the system of \cite{M}.

Except in the rational limit \cite{P1}, no Lax matrix is known that would generate van
Diejen's commuting Hamiltonians. In the reduction approach a Lax matrix arises
automatically, in our case it features in equations \eqref{F22} and \eqref{4.61}.
This might be helpful in searching for a Lax matrix behind van Diejen's 5-coupling Hamiltonian.
The search would be easy if one could derive van Diejen's system by Hamiltonian reduction.
It is a long standing open problem to find such derivation. Perhaps one should
consider some `classical analogue' of the quantum group interpretation of the
Koornwinder ($\BC_n$ Macdonald) polynomials found in \cite{OS}, since those
polynomials diagonalize van Diejen's quantized Hamiltonians \cite{vD3}.

Another open problem is to construct action-angle duals of the deformed $\BC_n$
Sutherland systems. Duality relations are not only intriguing on their own right,
but are also very useful for extracting information about the dynamics
\cite{Ruij88,RIMS95,RuijR,P2}.
The duality was used in \cite{AFG,FG1} to show that the hyperbolic $\BC_n$
Sutherland system is maximally superintegrable, the trigonometric $\BC_n$ Sutherland
system has precisely $n$ constants of motion, and the relevant dual systems are
maximally superintegrable in both cases. These studies, which were heavily influenced
by Pusztai's paper \cite{P1} (see also \cite{FG2}), may provide inspiration for a
future investigation of the dualities for the deformed $\BC_n$ Sutherland systems.
We here only remark that deformed dual systems should arise from considering the
reduction of alternative sets of commuting free Hamiltonians on the pertinent
Heisenberg doubles.

After we finished our work, there appeared a preprint \cite{vDE} dealing with the
quantum mechanics of a lattice version of a 4-parameter Inozemtsev type limit of
van Diejen's trigonometric/hyperbolic system. The systems studied in \cite{M} and
in our paper correspond to further limits of specializations of this one. The
statements about quantum mechanical dualities contained in \cite{vDE} and its
references should be related to classical dualities.

We hope be able to return to some of these questions in the future.

\begin{acknowledgements}
L.F. wishes to thank S. Ruijsenaars for suggesting the role of the kind of singular
limit described in Appendix \ref{sec:A}. This work was supported in part by the
Hungarian Scientific Research Fund (OTKA) under the grant K-111697. The work was
also partially supported by the European Union and the European Social Fund
through the project TAMOP-4.2.2.D-15/1 at the University of Szeged.
\end{acknowledgements}

\appendix

\section{Links to systems of van Diejen and Schneider}
\label{sec:A}

Recall that the trigonometric $\BC_n$ van Diejen system \cite{vD} has the Hamiltonian
\be
H_{\text{vD}}(\lambda,\theta)=\sum_{j=1}^n\big(\cosh(\theta_j)
\sV_j(\lambda)^{1/2}\sV_{-j}(\lambda)^{1/2}-[\sV_j(\lambda)+\sV_{-j}(\lambda)]/2\big),
\label{A.1}
\ee
with $\sV_{\pm j}$ ($j=1,\dots,n$) defined by
\be
\sV_{\pm j}(\lambda)=\sw(\pm\lambda_j)\prod_{\substack{k=1\\(k\neq j)}}^n
\sv(\pm\lambda_j+\lambda_k)\sv(\pm\lambda_j-\lambda_k),
\label{A.2}
\ee
and $\sv,\sw$ denoting the trigonometric potentials
\be
\sv(z)=\frac{\sin(\mu+z)}{\sin(z)}\quad\text{and}\quad
\sw(z)=\frac{\sin(\mu_0+z)}{\sin(z)}
\frac{\cos(\mu_1+z)}{\cos(z)}
\frac{\sin(\mu'_0+z)}{\sin(z)}
\frac{\cos(\mu'_1+z)}{\cos(z)},
\label{A.3}
\ee
where $\mu,\mu_0,\mu_1,\mu'_0,\mu'_1$ are arbitrary parameters.
By making the substitutions
\be
\begin{gathered}\lambda_j\to\ri(\hat p_j+R),\\\theta_j\to\ri\hat q_j,\end{gathered}
\quad\forall j\quad\text{and}\quad
\mu\to\ri g/2,\quad
\begin{gathered}\mu_0\to\ri(g_0+R),\\\mu'_0\to\ri(g'_0-R),\end{gathered}\quad
\begin{gathered}\mu_1\to\ri g_1+\pi/2,\\\mu'_1\to\ri g'_1+\pi/2\end{gathered}
\label{A.4}
\ee
the potentials become hyperbolic functions and their $R\to\infty$ limit exists, namely
\be
\lim_{R\to\infty}\sv(\pm (\lambda_j+\lambda_k))=e^{\pm g/2},\quad
\lim_{R\to\infty}\sv(\pm (\lambda_j-\lambda_k))
=\frac{\sinh(\pm g/2+\hat p_j-\hat p_k)}{\sinh(\hat p_j-\hat p_k)},\quad\forall j,k
\label{A.5}
\ee
and
\be
\lim_{R\to\infty}\sw(\pm\lambda_j)
=e^{g_0-g'_0\pm(g_1+g'_1)-2\hat p_j}-e^{\pm(g_0+g'_0+g_1+g'_1)},\quad\forall j.
\label{A.6}
\ee
In the $1$-particle case we have $V_{\pm}(\lambda)=\sw(\pm\lambda)$,
thus $H_{\text{vD}}$ takes the following form
\be
H_{\text{vD}}(\lambda,\theta)=\cosh(\theta)\sw(\lambda)^{1/2}\sw(-\lambda)^{1/2}
-[\sw(\lambda)+\sw(-\lambda)]/2.
\label{A.7}
\ee
By utilizing \eqref{A.6} one obtains
\be
\begin{gathered}
\lim_{R\to\infty}\sw(\lambda)^{1/2}\sw(-\lambda)^{1/2}
=\big[1-(e^{2g_0}+e^{-2g'_0})e^{-2\hat p}+e^{2g_0-2g'_0-4\hat p}\big]^{1/2},\\
\lim_{R\to\infty}[\sw(\lambda)+\sw(-\lambda)]/2
=\frac{e^{g_0-g'_0+g_1+g'_1}+e^{g_0-g'_0-g_1-g'_1}}{2}e^{-2\hat p}
-\cosh(g_0+g'_0+g_1+g'_1).
\end{gathered}
\label{A.8}
\ee
Equating the $R\to\infty$ limit of $H_{\text{vD}}(\lambda,\theta)$ \eqref{A.7} with
the Hamiltonian $H(\hat p,\hat q;x,u,v)$ \eqref{I1} yields a system of linear equations
involving $g_0,g_1,g'_0,g'_1$ as unknowns and $u,v$ as parameters. Actually, four
sets of linear equations can be constructed, each with infinitely many solutions
depending on one (real) parameter, but these sets are `equivalent' under the
exchanges: $g_0\leftrightarrow g'_0$ or $g_1\leftrightarrow g'_1$.
Therefore it is sufficient to give only one set of solutions, e.g.
\be
g_0=v-u,\quad g_0'=0,\quad g_1=u+v-g'_1,\quad g'_1\in\R.
\label{A.9}
\ee
Setting $g=x$ and $g'_1=0$ provides the following special choice of couplings in
\eqref{A.4}
\be
\mu=\ri x/2,\quad
\mu_0=\ri(v-u+R),\quad
\mu'_0=-\ri R,\quad
\mu_1=\ri(u+v)+\pi/2,\quad
\mu'_1=\pi/2,
\label{A.10}
\ee
and one finds the following
\be
\lim_{R\to\infty}H_{\text{vD}}\big(\lambda(\hat p,R),\theta(\hat q)\big)
=-H(\hat p,\hat q;x,u,v)+\cosh\big(2u\big).
\label{A.11}
\ee
In the $n$-particle case, by using \eqref{A.5} and \eqref{A.6} it can be shown
that with \eqref{A.10} one has
\be
\lim_{R\to\infty}H_{\text{vD}}\big(\lambda(\hat p,R),\theta(\hat q)\big)
=-H(\hat p,\hat q;x,u,v)+\sum_{j=1}^n\cosh\big((j-1)x+2u\big),
\label{A.12}
\ee
i.e., the Hamiltonian $H$ \eqref{I1} is recovered as a singular limit of $H_{\text{vD}}$
\eqref{A.1}.

Consider now the function $H(\hat p,\hat q;x,u,v)$
and introduce the real parameter $\sigma$ through the substitutions
\be
u\to u-\sigma,\quad v\to v-\sigma
\label{A.13}
\ee
and apply the canonical transformation
\be
\hat p_j\to -Q_j+\sigma,\quad \hat q_j\to -P_j,\quad\forall j.
\label{A.14}
\ee
Then we have
\be
\lim_{\sigma\to\infty}H(\hat p(Q,\sigma),\hat q(P),x,u(\sigma),v(\sigma))
=H_{\text{Sch}}(Q,P,x,u),
\label{A.15}
\ee
with Schneider's \cite{S} Hamiltonian
\be
H_{\text{Sch}}(Q,P,x,u)
=\frac{e^{-2u}}{2}\sum_{j=1}^ne^{2Q_j}-\sum_{j=1}^n\cos(P_j)
\prod_{\substack{k=1\\(k\neq j)}}^n
\bigg[1-\frac{\sinh^2\big(\frac{x}{2}\big)}
{\sinh^2(Q_j-Q_k)}\bigg]^{\tfrac{1}{2}}.
\label{A.16}
\ee

\begin{remark}\label{rem:A.1}
(\emph{i}) In \eqref{A.4} only two of the four external field couplings
$\mu_0,\mu_0',\mu_1,\mu_1'$ are scaled with $R$.
However, scaling all four of these parameters also leads to
an integrable Ruijsenaars-Schneider type system with a more
general $4$-parameter external field. For details, see Section
II.B of \cite{vD2}. (\emph{ii}) The connection to Schneider's Hamiltonian
was mentioned in Remark 7.1 of \cite{M} as well, where a singular limit,
similar to \eqref{A.15} was taken.
\end{remark}

\section{Proof of a key result}
\label{sec:B}

In this appendix we prove Proposition \ref{prop:3.2} which states that the range of the
`position variable' $\hat p$ is contained in the closed thick-walled Weyl chamber
$\bar\cC_x$ \eqref{S18}.

\begin{proof}[Proof of Proposition \ref{prop:3.2}]
According to \eqref{S17} the matrices $e^{2\hat p}$ and
$e^{2\hat p-x\1_n}+\sgn(x)e^{\hat p}ww^\dag e^{\hat p}$
are similar and therefore have the same characteristic polynomial.
This gives the identity
\be
\prod_{j=1}^n(e^{2\hat p_j}-\lambda)=\prod_{j=1}^n(e^{2\hat p_j-x}-\lambda)
+\sgn(x)\sum_{j=1}^n\bigg[e^{2\hat p_j}|w_j|^2\prod_{\substack{k=1\\(k\neq j)}}^n
(e^{2\hat p_k-x}-\lambda)\bigg],
\label{B1}
\ee
where $\lambda$ is an arbitrary complex parameter. The constraint on $\hat p$ arises
from the fact that $\vert w_m \vert^2$ $(m=1,\ldots, n)$ must be non-negative and not
all zero because of the definition \eqref{S16}.

Let us assume for a moment that the components of $\hat p$ are distinct such that
$\hat p_1>\dots>\hat p_n$. This enables us to express $|w_m|^2$ for all
$m\in \{1,\dots,n\}$ from the above equation by evaluating it at $n$ different values
of $\lambda$, viz. $\lambda=e^{2\hat p_m-x}$, $m=1,\dots,n$. We obtain the following
\be
|w_m|^2=\sgn(x)(1-e^{-x})\prod_{\substack{j=1\\(j\neq m)}}^n
\frac{e^{2\hat p_j+x}-e^{2\hat p_m}}{e^{2\hat p_j}-e^{2\hat p_m}},
\quad m=1,\dots,n.
\label{B2}
\ee
For $x>0$ and any $\hat p$ with $\hat p_1>\dots>\hat p_n$ the formula \eqref{B2}
implies that $|w_n|^2>0$ and for $m=1,\dots,n-1$ we have $|w_m|^2\geq 0$ if and only
if $\hat p_m-\hat p_{m+1}\geq x/2$. Similarly, if $x<0$ and $\hat p\in\R^n$ with
$\hat p_1>\dots>\hat p_n$, then \eqref{B2} implies $|w_1|^2>0$ and for $m=2,\dots,n$
we have $|w_m|^2\geq 0$ if and only if $\hat p_{m-1}-\hat p_m\geq -x/2$. In summary,
if $\hat p_1>\dots>\hat p_n$, then $|w_m|^2\geq 0$ $\forall m$ implies that
$\hat p\in\bar\cC_x$.

Now, let us prove our assumption, that all components of $\hat p$ must be different.
Indirectly, suppose that some (or maybe all) of the $\hat p_j$'s coincide.
This can be captured by a partition of the positive integer
\be
n=k_1+\dots+k_r,
\label{B3}
\ee
where $r<n$ (or equivalently, at least one integer $k_1,\dots,k_r$ must be
greater than $1$) and the indirect assumption can be written as
\be
\hat p_1=\dots=\hat p_{k_1},\quad
\hat p_{k_1+1}=\dots=\hat p_{k_1+k_2},\quad\dots,\quad
\hat p_{k_1+\dots+k_{r-1}+1}=\dots=\hat p_{k_1+\dots+k_r}\equiv\hat p_n.
\label{B4}
\ee
Then \eqref{B1} can be reformulated as
\be
\prod_{j=1}^r(\Delta_j-\lambda)^{k_j}
=\prod_{j=1}^r(\Delta_je^{-x}-\lambda)^{k_j}
+\sgn(x)\sum_{m=1}^rZ_m\Delta_m(\Delta_me^{-x}-\lambda)^{k_m-1}
\prod_{\substack{j=1\\(j\neq m)}}^r(\Delta_je^{-x}-\lambda)^{k_j},
\label{B5}
\ee
where we introduced $r$ distinct variables
\be
\Delta_1=e^{2\hat p_{k_1}},\quad
\Delta_2=e^{2\hat p_{k_1+k_2}},\quad\dots,\quad
\Delta_r=e^{2\hat p_{k_1+\dots+k_r}}\equiv e^{2\hat p_n},
\label{B6}
\ee
and $r$ non-negative real variables
\be
\begin{gathered}
Z_1=|w_1|^2+\dots+|w_{k_1}|^2,\quad
Z_2=|w_{k_1+1}|^2+\dots+|w_{k_1+k_2}|^2,\\
\dots,\quad Z_r=|w_{k_1+\dots+k_{r-1}+1}|^2+\dots+|w_n|^2.
\end{gathered}
\label{B7}
\ee
Notice that $Z_1+\dots+Z_r=|w|^2=\sgn(x)e^{-x}(e^{nx}-1)>0$, therefore at least one
of the $Z_j$'s must be positive. Next, we define the rational function of $\lambda$
\be
Q(\Delta,x,\lambda)=\prod_{j=1}^r\frac{(\Delta_j-\lambda)^{k_j}}
{(\Delta_je^{-x}-\lambda)^{k_j-1}},
\label{B8}
\ee
and use it to rewrite \eqref{B5} as
\be
Q(\Delta,x,\lambda)=\prod_{j=1}^r(\Delta_je^{-x}-\lambda)
+\sgn(x)\sum_{m=1}^rZ_m\Delta_m
\prod_{\substack{j=1\\(j\neq m)}}^r(\Delta_je^{-x}-\lambda).
\label{B9}
\ee
The above equation implies that all poles of $Q$ are apparent, i.e., there must
be cancelling factors in its numerator. This observation has a straightforward
implication on the $\Delta$'s.
\begin{center}
($\ast$)\quad For every index $m\in\{1,\dots,r\}$ with $k_m>1$,
there exists an index $s\in\{1,\dots,r\}$ s.t. $\Delta_s=\Delta_me^{-x}$
and $k_s\geq k_m-1$.
\end{center}
The quantities $Z_m=Z_m(\Delta,x)$ can be uniquely determined by evaluating
\eqref{B9} at $r$ different values of the parameter $\lambda$, namely
$\lambda_m=\Delta_me^{-x}$ ($m=1,\dots,r$). However, there are $3$ disjoint
cases which are to be handled separately.

\noindent
\underline{Case 1:} $k_m=1$ and $\nexists s\in\{1,\dots,r\}$:
$\Delta_s=\Delta_me^{-x}$. Then we find
\be
Z_m=\sgn(x)(1-e^{-x})e^{(n-1)x}\prod_{\substack{j=1\\(j\neq m)}}^r
\bigg(\frac{\Delta_j-\Delta_me^{-x}}{\Delta_j-\Delta_m}\bigg)^{k_j}>0.
\label{B10}
\ee
\underline{Case 2:} $k_m>1$ and $k_s=k_m-1$. Then we find
\be
Z_m=(-1)^{k_m+1}\sgn(x)(1-e^{-x})e^{(n-k_m)x}\prod_{\substack{j=1\\(j\neq m,s)}}^r
\bigg(\frac{\Delta_j-\Delta_me^{-x}}{\Delta_j-\Delta_m}\bigg)^{k_j}>0.
\label{B11}
\ee
\underline{Case 3:} $k_m=1$ and $\exists s\in\{1,\dots,r\}$:
$\Delta_s=\Delta_me^{-x}$ or $k_m>1$ and $k_s>k_m-1$. Then we get
\be
Z_m=0.
\label{B12}
\ee
Since there is at least one $Z_m$ which is positive, the set of indices belonging to
Case 1 or Case 2 must be non-empty. Introduce a real positive parameter $\varepsilon$
and associate to every degenerate configuration \eqref{B4} a continuous
family of configurations, denoted by $\hat p(\varepsilon)$, with components
$\hat p(\varepsilon)_1,\dots,\hat p(\varepsilon)_n$ defined by the formulae
\be
\begin{gathered}
\exp(2\hat p(\varepsilon)_a+a\varepsilon)=\Delta_1,\quad a=1,\dots,k_1,\\
\exp(2\hat p(\varepsilon)_{\sum_{m=1}^{j-1}k_m+a}+a\varepsilon)=\Delta_j,
\quad a=1,\dots,k_j,\quad j=2,\dots,r.
\end{gathered}
\label{B13}
\ee
This way coinciding components of $\hat p$ \eqref{B4} are `pulled apart' to points
successively separated by $\varepsilon/2$. It is clear that with sufficiently small
separation the configuration $\hat p(\varepsilon)$ sits in the chamber
$\{\hat x\in\R^n\mid 0>\hat x_1>\dots>\hat x_n\}$. For such non-degenerate
configurations $\hat p(\varepsilon)$, let us consider the expressions
\be
|w_\ell(\hat p(\varepsilon),x)|^2=\sgn(x)(1-e^{-x})
\prod_{\substack{j=1\\(j\neq\ell)}}^n
\frac{e^{2\hat p(\varepsilon)_j+x}-e^{2\hat p(\varepsilon)_\ell}}
{e^{2\hat p(\varepsilon)_j}-e^{2\hat p(\varepsilon)_\ell}},\quad\ell=1,\dots,n,
\label{B14}
\ee
which give the unique solution of equation \eqref{B1} at $\hat p(\varepsilon)$.
The limits $\lim_{\varepsilon\to 0}|w_\ell(\hat p(\varepsilon),x)|^2$ exist,
and do not vanish for $\ell=k_1+\dots+k_m$ if $k_m$ belongs to Case 1 or Case 2.
For such $\ell=k_1+\dots+k_m$ we must have
\be
\lim_{\varepsilon\to 0}|w_{k_1+\dots+k_m}(\hat p(\varepsilon),x)|^2=Z_m(\Delta,x)>0,
\label{B15}
\ee
where $Z_m$ is given by \eqref{B10} in Case 1 and by \eqref{B11} in Case 2.
It can be also seen that
\be
|w_\ell(\hat p(\varepsilon),x)|^2\equiv 0
\quad\Longleftrightarrow\quad
\begin{cases}
\ell\notin\{k_1,k_1+k_2,\dots,k_1+\dots+k_r\}\\
\text{or}\\
\ell=k_1+\dots+k_m\ \text{with}\ k_m\ \text{from Case 3},
\end{cases}
\label{B16}
\ee
i.e., $|w_\ell(\hat p(\varepsilon),x)|^2$ vanishes identically except for the
components in \eqref{B15}. Notice that for a small enough $\varepsilon$ some
coordinates of $\hat p(\varepsilon)$ are separated by less than $|x|/2$. Thus,
as it was shown at beginning the proof, we have $|w_\ell(\hat p(\varepsilon),x)|^2<0$
for some index $\ell$, which might depend on $\varepsilon$. Moreover, \eqref{B16}
implies that the index in question must have the form
$\ell=k_1+\dots+k_{m^\ast}$ for some $m^\ast$ appearing in \eqref{B15}.
But since the number of indices is finite, a monotonically decreasing sequence
$\{\varepsilon_N\}_{N=1}^\infty$ tending to zero can be chosen such that
$|w_{k_1+\dots+k_{m^\ast}}(\hat p(\varepsilon_N),x)|^2<0$ for all $N$.
This together with \eqref{B16} gives the contradiction
\be
0\geq \lim_{N\to\infty}|w_{k_1+\dots+k_{m^\ast}}(\hat p(\varepsilon_N),x)|^2
=Z_{m^\ast}(\Delta,x) >0
\label{B17}
\ee
proving that all components of $\hat p$ must be distinct. This concludes the proof.
\end{proof}

The above proof is a straightforward adaptation of the proofs of Lemma 5.2 of \cite{FK1}
and Theorem 2 of \cite{FK2}. We presented it since it could be awkward to extract the
arguments from those lengthy papers, and also our notations and the ranges of our
variables are different.

\section{Proof of an elementary lemma}
\label{sec:C}

We here prove the following equivalent formulation of Lemma \ref{lem:4.1}.

\begin{lemma}\label{lem:C.1}
Suppose that $\frac{\pi}{2}\geq q_1>\dots>q_n>0$ and
\be
\begin{bmatrix}\eta_L(1)&\0_n\\\0_n&\eta_L(2)\end{bmatrix}
\begin{bmatrix}\cos q&\ri\sin q\\\ri\sin q&\cos q\end{bmatrix}
\begin{bmatrix}\eta_R(1)^{-1} &\0_n\\\0_n&\eta_R(2)^{-1}\end{bmatrix}
=\begin{bmatrix}\cos q&\ri\sin q\\\ri\sin q&\cos q\end{bmatrix}
\label{C1}
\ee
for $\eta_L,\eta_R\in G_+$. Then
\be
\eta_L(1)=\eta_R(2)=m_1,\quad\eta_L(2)=\eta_R(1)=m_2
\label{C2}
\ee
with some diagonal matrices $m_1,m_2\in\T_n$ having the form
\be
m_1=\diag(a,\xi),\quad m_2=\diag(b,\xi),\quad\xi\in\T_{n-1},\
a,b\in\T_1,\quad\det(m_1 m_2)=1.
\label{C3}
\ee
If in addition $\frac{\pi}{2}>q_1$, then $m_1=m_2$.
\end{lemma}

\begin{proof}
The block off-diagonal components of the equality \eqref{C1} give
\be
\eta_L(1)=(\sin q)\eta_R(2)(\sin q)^{-1},\quad
\eta_L(2)=(\sin q)\eta_R(1)(\sin q)^{-1}.
\label{C4}
\ee
Since $\eta_L(1)^{-1}=\eta_L(1)^\dag$, the first of these relations implies
$\eta_R(2)=(\sin q)^2\eta_R(2)(\sin q)^{-2}$. As the entries of $(\sin q)$ are all
different, this entails that $\eta_R(2)$ is diagonal, and consequently we obtain the
relations in \eqref{C2} with some diagonal matrices $m_1$ and $m_2$. On the other
hand, the block-diagonal components of \eqref{C1} require that
\be
\cos q =\eta_L(1)(\cos q)\eta_R(1)^{-1},\quad
\cos q =\eta_L(2)(\cos q)\eta_R(2)^{-1}.
\label{C5}
\ee
Since $\cos q_k\neq 0$ for $k=2,\dots,n$, the formula \eqref{C3} follows.
If an addition $\cos q_1\neq 0$, then we also obtain from \eqref{C5} that $a=b$,
i.e., $m_1=m_2=m$ with some $m\in\T_n$.
\end{proof}

\section{Auxiliary material on Poisson-Lie symmetry}
\label{sec:D}

The statements presented here are direct analogues of well-known results \cite{AMM81,GS}
about Hamiltonian group actions with zero Poisson bracket on the symmetry group.
They are surely familiar to experts, although we could not find them in a reference.

Let us consider a Poisson-Lie group $G$ with dual group $G^\ast$ and a symplectic
manifold $P$ equipped with a left Poisson action of $G$. Essentially following Lu
\cite{Lu} (cf.~Remark \ref{rem:D.4}), we say that the $G$-action admits the momentum map
$\psi\colon P\to G^\ast$ if for any $X\in\cG$, the Lie algebra of $G$, and any
$f\in C^\infty(P)$ we have
\be
(\cL_{X_P}f)(p)=\langle X,\{f,\psi\}(p)\psi(p)^{-1}\rangle,\quad\forall p\in P,
\label{D.1}
\ee
where $X_P$ is the vector field on $P$ corresponding to $X$, $\langle.,.\rangle$
stands for the canonical pairing between the Lie algebras of $G$ and $G^\ast$, and
the notation pretends that $G^\ast$ is a matrix group. Using the Hamiltonian vector
field $V_f$ defined by $\cL_{V_f}h=-\{f,h\}$ ($\forall h\in C^\infty(P)$), we can
spell out equation \eqref{D.1} equivalently as
\be
(\cL_{X_P}f)(p)=-\langle X,\big(D_{\psi(p)}R_{\psi(p)^{-1}}\big)
\big((D_p\psi)(V_f(p))\big)\rangle,\quad\forall p\in P,
\label{D.2}
\ee
where $D_p\psi\colon T_p P\to T_{\psi(p)}G^\ast$ is the derivative, and
$R_{\psi(p)^{-1}}$ denotes the right-translation on $G^\ast$ by $\psi(p)^{-1}$.
Since the vectors of the form $V_f(p)$ span $T_p P$, we obtain the following
characterization of the Lie algebra of the isotropy subgroup $G_p<G$ of $p\in P$.

\begin{lemma}\label{lem:D.1}
With the above notations, we have
\be
\mathrm{Lie}(G_p)=\big[\big(D_{\psi(p)}R_{\psi(p)^{-1}}\big)
\big(\mathrm{Im}(D_p\psi)\big)\big]^\perp.
\label{D.3}
\ee
\end{lemma}

This directly leads to the next statement.

\begin{corollary}\label{cor:D.2}
An element $\mu\in G^\ast$ is a regular value of the momentum
map $\psi$ if and only if $\mathrm{Lie}(G_p)=\{0\}$ for every
$p\in\psi^{-1}(\mu)=\{p\in P\mid\psi(p)=\mu\}$.
\end{corollary}

Let us further suppose that $\psi\colon P\to G^\ast$ is $G$-equivariant, with respect
to the appropriate dressing action of $G$ on $G^\ast$. Then we have
\be
G_p<G_\mu,\quad\forall p\in\psi^{-1}(\mu).
\label{D.4}
\ee
Here $G_p$ and $G_\mu$ refer to the respective actions of $G$ on $P$ and on $G^\ast$.
Corollary \ref{cor:D.2} and equation \eqref{D.4} together imply the following useful result.

\begin{corollary}\label{cor:D.3}
If $G_\mu$ acts locally freely on $\psi^{-1}(\mu)$, then $\mu$
is a regular value of the equivariant momentum map $\psi$. Consequently,
$\psi^{-1}(\mu)$ is an embedded submanifold of $P$.
\end{corollary}

We finish by a clarifying remark concerning the momentum map.

\begin{remark}\label{rem:D.4}
Let $B$ be the Poisson tensor on $P$, for which $\{f,h\}=B(df,dh)=\cL_{V_h}f$.
We can write $V_h=B^\sharp(dh)$ with the corresponding bundle map
$B^\sharp\colon T^\ast P\to TP$. Any $X\in\cG=T_eG=(T_{e'}G^\ast)^\ast$ extends to a
unique right-invariant 1-form $\vartheta_X$ on $G^\ast$ ($e\in G$ and $e'\in G^\ast$ are
the unit elements). With this at hand, equation \eqref{D.1} can be reformulated as
\be
X_P=B^\sharp(\psi^\ast(\vartheta_X)),
\label{D.5}
\ee
which is a slight variation of the defining equation of the momentum map found in \cite{Lu}.
\end{remark}

\end{document}